\documentclass{article}

 \usepackage[preprint]{neurips_2026}

\usepackage[utf8]{inputenc} 
\usepackage[T1]{fontenc}    
\usepackage{hyperref}       
\usepackage{url}            
\usepackage{booktabs}       
\usepackage{amsfonts}       
\usepackage{nicefrac}       
\usepackage{microtype}      
\usepackage{xcolor}         
\usepackage{amsmath}
\usepackage{algorithm}
\usepackage{algpseudocode}
\usepackage{amsthm}
\usepackage{thmtools}
\usepackage{thm-restate}
\usepackage{graphicx}
\usepackage{caption}
\usepackage{enumitem}

\usepackage{wrapfig}

\newtheorem{theorem}{Theorem}[section]
\newtheorem{corollary}{Corollary}[theorem]
\newtheorem{lemma}[theorem]{Lemma}

\newtheorem{definition}[theorem]{Definition}

\DeclareMathOperator*{\argmin}{argmin}
\DeclareMathOperator*{\argmax}{argmax}

\title{Robust Simulation Based Inference Through Robust Optimal Transport}

\author{%
  Peter Matthew Jacobs \\
  Department of Statistics\\
  University of Wisconsin-Madison\\
  Madison, WI 53706 \\
\texttt{pjacobs5@wisc.edu} \\
\And
Lekha Patel \\
  Scientific Machine Learning Department\\
  Sandia National Laboratories\\
  Albuquerque, NM 87185\\
\texttt{lpatel@sandia.gov}
\And
Anirban Bhattacharya \\
Department of Statistics \\
Texas A\&M University \\
College Station, TX 77843 \\
\texttt{anirbanb@stat.tamu.edu}
\And
Debdeep Pati\\
Department of Statistics \\
University of Wisconsin-Madison\\
Madison, WI 53706 \\
\texttt{dpati2@wisc.edu}
}

\begin{document}

\maketitle

\begin{abstract}
  When a statistical model $\{P_{\theta} : \theta \in \Theta\}$ lacks analytically tractable likelihoods, parametric statistical inference based on data generated from an unknown underlying distribution $P$ can still be performed as long as simulations from the model are possible. This approach is called Simulation Based Inference (SBI). Statistical models are rarely exactly correct (that is, $P \notin \{P_{\theta}: \theta \in \Theta\}$), and Robust SBI focuses on inferring a reasonable parameter even under model mis-specification. We focus on the setting where $P$ possesses potentially both geometric and Total Variation type discrepancies from $P_{\theta^*}$. For this problem, we use a Kullback-Liebler informed robust Optimal Transport divergence, motivated by Empirical Likelihood considerations. We introduce a stochastic sub-gradient ascent algorithm with a convergence guarantee for estimating the semi-discrete version of this robust Optimal Transport divergence, and design a parallelized SBI algorithm which employs the regular bootstrap on top of minimum semi-discrete robust Optimal Transport for parameter uncertainty quantification. We demonstrate mathematically why the divergence is robust under a joint geometric plus Total Variation type contamination and then illustrate the robustness of inferences on a complex benchmark SBI task.
\end{abstract}

\section{Introduction}
\label{sec:intro}

Consider the scenario where a practitioner selects a  statistical model $\{P_{\theta} : \theta \in \Theta\}$ and would like to infer the parameters based on a collection of independent and identically distributed data $Y_1,\dots,Y_n \overset{iid}{\sim} P$. Standard statistical inference makes the often unrealistic assumption the data generating probability distribution $P$ is a member of the model (i.e $P = P_{\theta^{*}}$ for some $\theta^{*} \in \Theta$). Contamination models allow to study the robustness properties of a statistical method. For example the standard Huber Contamination model \citep{huber1992robust} assumes $P = (1-\epsilon) P_{\theta^{*}} + \epsilon F$ for an arbitrary distribution $F$ and some small $\epsilon > 0$. A method is called robust to a contamination if a small amount of it, for example small $\epsilon$ in the Huber model, cannot lead to arbitrarily large changes in the estimate for $\theta^{*}$. The Huber contamination assumption is rigid in that it does not allow for the introduction of a geometric perturbation to $P_{\theta^{*}}$. Consider for example the case that with probability $1-\epsilon$, a sample is drawn from $P_{\theta^{*}}$, but then is moved a distance at most $\rho$. Specifically, a sample $Y \sim P$ if $X \sim P_{\theta^{*}},W \sim bernoulli(\epsilon),Z \sim F$ are independent and $Y = (1-W) T(X) + W Z$ where $|T(x) -x| \leq \rho$ for all $x$. This type of \textit{all point} contamination is well captured by the Wasserstein distance. Letting $\mathcal{P}_{p}(\mathbb{R}^{m})$ be the Borel probability measures on $\mathbb{R}^{m}$ with finite $p^{th}$ moment for $p \geq 1$, the Wasserstein-$p$ distance \citep{villani2009optimal,chewi2025statistical} is defined as $
W_p(P,Q) = \left( \inf_{\pi \in \Pi(P,Q)} \mathbb{E}_{(A,B) \sim \pi} \| A - B\|_2^{p}\right)^{1/p}$ where $P,Q \in \mathcal{P}_{p}(\mathbb{R}^{m})$ and $\Pi$ is the set of couplings between $P$ and $Q$ (joint distributions whose marginals are constrained to be $P$ and $Q$ respectively). A much more general contamination model encompassing this deterministic error plus Huber Contamination scenario supposes only that $P \in \{\nu: \exists Q, \mathrm{TV}(Q,\nu) \leq \epsilon, W(Q,P_{\theta^{*}})\leq \rho\}$ where $\mathrm{TV}$ stands for Total Variation distance \citep{gibbs2002choosing}. We consider a tractable and interpretable subclass of this general contamination model which assumes $P = (1-\epsilon)Q + \epsilon F$ for some distributions $F$ and $Q$ such that $W_2(Q,P_{\theta^{*}}) \leq \rho$. $\rho = 0$ reduces to the Huber Contamination model, while $\epsilon = 0$ reduces to only geometric contamination. We call this the Geometric+Huber (\textbf{G+H}) contamination model, and it captures both mass contamination and geometric perturbations.



\begin{wrapfigure}{r}{0.47\textwidth}
    \centering
\includegraphics[width=0.38\textwidth]{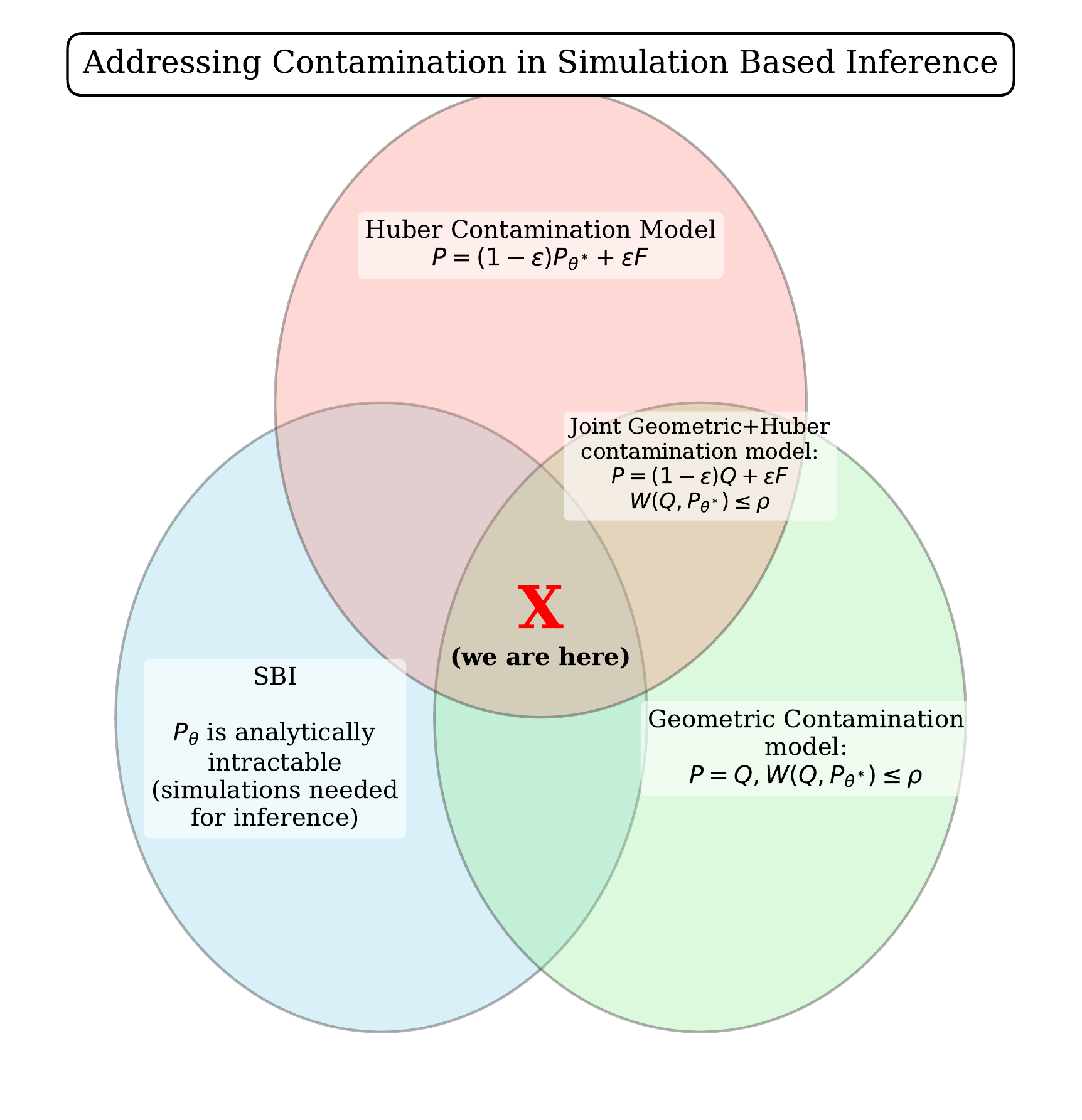}
    \caption{In iid statistical inference, there is a statistical model, $\{P_{\theta} : \theta \in \Theta\}$, with the measures defined on a ground space $\mathcal{Y}$. Data $Y_1,\dots,Y_n \overset{iid}{\sim} P$, and the parameters need to be inferred from the data.}
    \label{fig:my_wrapped_image}
\end{wrapfigure}

However, robustness to multiple forms of contamination is not the only challenge faced by the practitioner. Another common challenge is that the  likelihood functions associated to distributions in the model are analytically intractable (i.e not available in closed form), but it is possible to simulate random vectors from $P_{\theta}$ for any $\theta$ by running a potentially complex and expensive computer code. Specifically, there exists a reference measure $\mu$ supported on $\mathbb{R}^{z}$ for some $z \geq 1$ and a function $G: \Theta \times \mathbb{R}^{z}$ such that for any $\theta \in \Theta$, if $X \sim \mu$, then $G(\theta,X) \sim P_{\theta}$. Performing inference by generating simulations from  members of the statistical model using the potentially expensive computer code $G$ is the Simulation Based Inference (SBI) \citep{brehmer2022simulation,gourieroux1993simulation,wang2024comprehensive} setting (also sometimes known as the Likelihood Free setting), and it arises in many important domains, such as Neuroscience, Epidemiology, Robotics, and Ecology \citep{nelson1998hodgkin,wang2016predicting,marlier2023simulation,wangersky1978lotka}.

In complex SBI applications, immense effort is put into constructing simulators based on domain science to reflect state of the art scientific knowledge. Yet, they are approximations to reality, and therefore it is important to account for departures from the simulator in a flexible manner. This motivates the study of more general contamination models in the SBI context. We introduce a method that can jointly address these challenges. Namely robustness under the G+H contamination model in addition to utilizable in the SBI setting.

Our inference method is a minimum divergence approach \citep{basu2011statistical} integrated with the regular bootstrap \citep{efron1987better} for quantification of uncertainty. Specifically, we define a robust divergence $\ell$ between probability distributions and target the parameter achieving the population minimum loss $\argmin_{\theta \in \Theta} \ell(P,P_{\theta})$, which serves as a robust surrogate for $\theta^{*}$ in the G+H Contamination model. Since $P$ is unknown but iid samples $Y_1,\dots,Y_n \overset{iid}{\sim} P$ are available, the minimum divergence estimator is $\argmin_{\theta \in \Theta} \ell(P_n,P_{\theta})$ where $P_n := \sum_{j=1}^{n} \frac{1}{n} \delta_{Y_j}$ and where $\delta_{x}$ refers to the Dirac measure at position $x$. Given a number of bootstrap samples $M$, to produce measures of uncertainty for the underlying parameter, for $j \in \{1,2,\dots, M\}$, a new dataset $Y'_{1,j},\dots,Y'_{n,j}$ is constructed by sampling uniformly and with replacement from $Y_1,\dots,Y_n$. The $j^{th}$ bootstrap sample for the parameter estimate is $\argmin_{\theta \in \Theta} \ell(P^{j}_n,P_{\theta})$ where $P^{j}_n$ is the empirical measure of the $j^{th}$ bootstrapped dataset.
\paragraph{Divergence selection}
Although it is geometrically descriptive, Wasserstein distance is not robust to Huber contamination. Consider for example attempting to infer the mean of a Normal distribution. If $P = (1-\epsilon) Normal(\theta^{*},1)+\epsilon Normal(\theta_1,1)$, then $\theta^{*} (1-\epsilon) + \theta_1 \epsilon = \argmin_{\theta \in \mathbb{R}} W_2( P, Normal(\theta,1) )$. Taking $\theta_1 \to \infty$, the minimizer can be made arbitrarily far from $\theta^{*}$ for any $\epsilon >0$. One way to robustify Wasserstein distance to Huber Contamination is to open a small ball $B_{\epsilon}(P)$ around $P$, and define a modified Wasserstein distance as $\inf_{Q \in B_{\epsilon}(P)} W(Q,Normal(\theta,1))$. The sentiment is that if $P$ is a Huber contamination of $P_{\theta^{*}}$, then $Normal(\theta^{*},1)$ will be in the ball. This is the motivation behind the following version of Robust Optimal Transport, which can be viewed as a soft-dual to the optimization $\inf_{Q:  \mathrm{KL}(Q,P) \leq \epsilon} W(Q,P_{\theta})$ where $\mathrm{KL}$ denotes the Kullback-Leibler divergence \citep{gibbs2002choosing}.

\begin{definition}[$\lambda$-Robust Semi-Constrained Wasserstein-2 ($\lambda$-RSW)]
    \label{def:RSW2}
    For probability measures $P_1$ and $P_2$ in $\mathcal{P}_{2}(\mathbb{R}^{m})$ and $\lambda > 0$,
\vspace{-1pt}
\begin{equation}\ell_{\lambda}(P_1,P_2) :=  \inf_{Q \ll P_1} \left[ \frac{1}{\lambda} \mathrm{KL}(Q,P_1)+W_2^2(P_2,Q) \right]
\end{equation}
\vspace{-1pt}
where $Q \ll P_1$ means $Q$ is absolutely continuous with respect to $P_1$.
\end{definition}
We select $\ell_{\lambda}$ as our divergence, motivated by recent developments in distributionally constrained Empirical Likelihood \citep{owen2001empirical,chakraborty2025robust}; see Appendix Section \ref{sec:EmpiricalLikelihoodMotivation} for more details on that motivation. Although motivating it in an entirely different context (via its connection to Unbalanced Optimal Transport \citep{chizat2017unbalanced}), \citet{le2021robust} appears to be the first to introduce this $\mathrm{KL}$ version of Robust Semi-Constrained Optimal Transport, where it is used in Barycenter computation. While the Wasserstein distance is defined through a infima over couplings of $P_1,P_2$, where the coupling by definition has a hard constraint that its first marginal is $P_1$ and its second marginal is $P_2$, $\ell_{\lambda}(P_1,P_2) = \inf_{Q \ll P_1} \left[ \frac{1}{\lambda} \mathrm{KL}(Q,P_1) + \inf_{\pi \in \Pi(P_2,Q)} \mathbb{E}_{(A,B) \sim \pi} \| A- B\|_2^2 \right]$ relaxes the first marginal constraint, hence it is called Semi-Constrained. We now define the $\lambda$-RSW parameter and its estimator.

\begin{definition}[$\lambda$-RSW Parameter]
    \label{def:MinDivergenceEstimator}
    For a given probability distribution $P$ and a statistical model $\{ P_{\theta}: \theta \in \Theta\}$, all probability measures defined on $\mathcal{P}_{2}(\mathbb{R}^{m})$, and a small $\eta > 0$, a $\lambda$-Robust Parameter is any $\theta_{\lambda,\eta} \in \Theta$ satisfying $\ell_{\lambda}(P,P_{\theta_{\lambda,\eta}}) \leq \inf_{\theta \in \Theta} \ell_{\lambda}(P,P_{\theta})+\eta$. Given data $Y_1,\dots,Y_n$, a Minimum $\lambda$-RSW Divergence estimator is any $\hat{\theta}_{\lambda,\eta} \in \Theta$ such that $ \ell_{\lambda}(P_n,P_{\hat{\theta}_{\lambda,\eta}}) \leq \inf_{\theta \in \Theta} \ell_{\lambda}(P_n,P_{\theta})+\eta$.
\end{definition}
$\eta > 0$ is in the definition only to recognize the pathological scenario where the infima is not achieved. The parameter and estimator are of interest when $\eta \approx 0$, and so for notational convenience in the remainder of this work, unless it is needed for mathematical clarity, we drop $\eta$ from the notation, referring instead to a $\lambda$-Robust parameter as $\theta_{\lambda}$ and a minimum $\lambda$-RSW Divergence estimator as $\hat{\theta}_{\lambda}$. These values are understood to correspond to some fixed small $\eta$.

\paragraph{Our contributions}

Our primary contribution is the development of a tool-kit for SBI under G+H Contamination via \underline{B}ootstrapping \underline{M}inimum \underline{R}obust \underline{S}emi-constrained \underline{W}asserstein-2 (B-MRSW). In the process of achieving this, we make the following contributions.
\begin{enumerate}[leftmargin=*,itemsep=0pt, topsep=2pt, parsep=0pt, partopsep=0pt]
    \item In Section \ref{sec:basicProperties}, we present mathematical insights about the $\lambda \to 0$ and $\lambda \to \infty$ behavior of the $\lambda$-RSW divergence.
    \item In Section \ref{sec:basicProperties}, we show that under the G+H Contamination setting where $P = (1-\epsilon) Q + \epsilon F$ and $W_2(P_{\theta^{*}},Q) \leq \rho$, $\ell_{\lambda}(P,P_{\theta^{*}})$ must be near to $\inf_{\theta \in \Theta} \ell_{\lambda}(P,P_{\theta})$ for intermediate $\lambda$ values, indicating minimum $\lambda$-RSW Divergence estimation is a reasonable approach to handling the G+H Contamination assumption. In particular, by targeting $\theta_{\lambda}$ for the appropriate $\lambda$, we can target $\theta^{*}$ (up to its identifiability).
    \item In Section \ref{sec:algorithms}, we provide a full algorithm for approximating the Minimum $\lambda$-RSW Divergence estimator given only the ability to simulate from members $P_{\theta}$ of the statistical model. The algorithm relies on computing the Semi-Discrete $\lambda$-RSW $(\ell_{\lambda}(P_n,P_{\theta}))$ and we develop a stochastic sub-gradient ascent algorithm for this purpose. In the process, we identify a dual formulation of $\lambda$-RSW arising from a natural interplay between entropy and Optimal Transport. We derive and highlight this dual formulation, which may be of independent interest.
    \item In Section \ref{sec:lamSelect}, we provide a principled, data-driven heuristic strategy for choosing $\lambda$, thereby completing the algorithmic pipeline for robust inferences under the G+H Contamination setting via Bootstrapped Minimum Robust Semi-constrained Wasserstein-2 Estimation.
    \item In Section \ref{sec:convergenceProperties} we provide a convergence guarantee for the stochastic sub-gradient ascent algorithm for computing Semi-Discrete $\lambda$-RSW, achieving $O(\epsilon)$ expected absolute error (with respect to the sampling from $P_{\theta}$) in approximating $\ell_{\lambda}(P_n,P_{\theta})$ in $O(\frac{n^2}{\epsilon^2})$ time. We believe this result will be of independent interest in the robust computational Optimal Transport domain.
    \item In Section \ref{sec:numericalIllustrations}, on a difficult benchmark task (inferring the parameters of the g-and-k distribution), we compare our approach to the NPL-MMD SBI method of \cite{dellaporta2022robust}, illustrating that robustness of Maximum Mean Discrepancy using the Gaussian Kernel is highly sensitive to small changes in choice of bandwidth, whereas using our data-driven heuristic for $\lambda$ selection, our highly parallelizable Bootstrapped, Minimum $\lambda$-RSW approach is robust in a large range of intermediate $\lambda$ values determined by the heuristic.
\end{enumerate}

We emphasize that our approach is designed to both target $\theta^{*}$ (up to its identifiability), while also providing confidence sets with the proper asymptotic coverage (via using the regular bootstrap).

\paragraph{Related work}
\label{sec:relatedWork}

Robust parametric statistics deals with the general setting where the parameters of a statistical model $\{P_{\theta}: \theta \in \Theta\}$ are to be inferred from data generated according to a probability distribution $P \notin \{P_{\theta}: \theta \in \Theta\}$, but $P$ is in some sense $\epsilon$ perturbed from $P_{\theta^{*}}$. A statistical method is robust if the small $\epsilon$ perturbation to the data generating distribution cannot be made to cause massive changes in inferences. What is meant by a small perturbation is made explicit using a contamination model \citep{huber1992robust,diakonikolas2019robust,liu2023robust}. We are not aware of works that point out the utility of a method for simultaneously addressing a geometric and Huber contamination.

 SBI is ultimately concerned with providing uncertainty quantification for the parameter of a statistical model, and this can either be accomplished in a Bayesian or Frequentist way. Notable Bayesian thrusts include Approximate Bayesian Computation \citep{beaumont2010approximate,beaumont2002approximate,blum2010approximate,beaumont2002approximate,lintusaari2017fundamentals,bernton2019approximate}, and Neural Network modeling of the likelihood \citep{papamakarios2019sequential} (which can be used in a frequentist or Bayesian inference framework) or posterior \citep{papamakarios2016fast,greenberg2019automatic}. A central question in robust statistics is how can one simultaneously estimate the proper quantity, while also providing confidence sets with proper asymptotic coverage. The $n \to \infty$ target of Maximum Likelihood Estimation is $\inf_{\theta \in \Theta} \mathrm{KL}(P,P_{\theta})$, which is not robust to Huber Contamination. Bayesian inference, being based on the likelihood, inherits this problem, but has the additional problem that posterior credible intervals can fail to have the correct coverage \citep{kleijn2012bernstein} around this non-robust target.
 \citet{miller2019robust} introduce the Coarsened Posterior to address the improper coverage of Bayesian inference, although this does not address the point estimation limitation of likelihood based inference (the non-robust target problem). To address this, a popular generalization of Bayesian inference used in the SBI community is the generalized posterior $p(\theta | X_1,\dots,X_n) \propto \exp(-\beta n \mathcal{L}(P_n,P_{\theta})) \pi(\theta)$ \citep{bissiri2016general,lyddon2019general}, where $\mathcal{L}$ is a robust empirical loss and $\beta$ is a user specified parameter. $\beta =1$ and $\mathcal{L}(P_n,P_{\theta}) = -\frac{1}{n} \sum_{i=1}^{n} \log p_{\theta}(X_i)$ recovers the original, non-robust posterior distribution (which targets $\mathrm{KL}(P,P_{\theta})$). Methods in this space utilizing Maximum Mean Discrepancy (MMD) \citep{li2017mmd} include \cite{cherief2020mmd} and \cite{pacchiardi2024generalized}. \cite{bharti2026amortised} uses a weighted score matching divergence. \cite{dellaporta2022robust} use minimum squared MMD within a non-parametric likelihood scheme which reduces to the Bayesian Bootstrap \citep{rubin1981bayesian} on top of minimum squared MMD when there is no prior. Robustness to Huber contamination is studied by \cite{cherief2020mmd} and \cite{dellaporta2022robust}, but in practice depends on careful selection of the kernel, which is a difficult functional parameter to tune. From the frequentist perspective \cite{bernton2019parameter} use the regular Bootstrap on top of vanilla, non robust, minimum Wasserstein distance.

The quantity $\ell_{\lambda}(P_1,P_2)$ we use for our min divergence approach is exactly the Robust Semi-constrained Optimal Transport introduced by \cite{le2021robust}. They focus on its application to Barycenter Computation for discrete distributions. They provide an algorithm for computing $\ell_{\lambda}(P_1,P_2)$ up to $\epsilon$ precision in $O(\frac{n^2}{\epsilon})$ time when $P_1,P_2$ are discrete; note that we address semi-discrete computation, which is naturally relevant in parametric statistical inference, where it is necessary to compare $P_n$ (discrete) with $P_{\theta}$ (which is often absolutely continuous with respect to Lebesgue measure). We also note that there are other ways of defining Robust Optimal Transport based on Total Variation \citep{mukherjee2021outlier,nietert2023robust,nietert2023outlier} or $\chi$-squared \citep{balaji2020robust} penalties. See our extended Related Work (Appendix Section \ref{sec:ExtendedRelatedWork}) for additional details.

\section{Properties of $\ell_{\lambda}$ and connections to robust inference}
\label{sec:basicProperties}
$\ell_{\lambda}$ is a divergence between probability measures in the following sense. 
\begin{restatable}{lemma}{divergence}
    \label{prop:positiveDefinite}
    $\ell_{\lambda}(P_1,P_2) \geq 0$ for every $P_1,P_2$ and $\ell_{\lambda}(P_1,P_2) = 0$ if and only if $P_1 = P_2$ for every $P_1,P_2 \in \mathcal{P}_{2}(\mathbb{R}^{m})$.
\end{restatable}
Regarding small $\lambda$ behavior, we have the following lemma.
\begin{restatable}{lemma}{smallLam}($\lambda$ small limit)
    \label{prop:smallLambda}
    Suppose $P_1,P_2 \in \mathcal{P}_{2}(A)$ where $A \subseteq \mathbb{R}^{m}$ with finite diameter. Then $\lim_{\lambda \to 0}  {\ell}_{\lambda}(P_1,P_2) = W_2^2(P_1,P_2)$.
\end{restatable}
For very small $\lambda$, we thus expect the Minimum $\lambda$-RSW estimator ( $\hat{\theta}_{\lambda}$) will behave like the non robust minimum Wasserstein-2 estimator. We also study the large $\lambda$ behavior in the semi-discrete setting.




\begin{restatable}{lemma}{largeLamDiscrete}(Large $\lambda$ Empirical Behavior)
    \label{largeLambdaEmpiricalBehavior}
    Let $y_1,y_2,\dots,y_n \in \mathbb{R}^{m}$, $P_2 \in \mathcal{P}_2(\mathbb{R}^{m})$ absolutely continuous with respect to Lebesgue.
    Then $\lim_{\lambda \to \infty}  {\ell}_{\lambda}(P_n,P_2) = \mathbb{E}_{X \sim P_2} \min_{j \in [n]} \| X - y_j \|_2^2$ where $P_n := \sum_{i=1}^{n} \frac{1}{n} \delta_{y_i}$.
\end{restatable}
The average nearest neighbor distance, $\mathbb{E}_{X \sim P_{\theta}} \min_{j \in [n]} \|X - y_j \|_2^2$, is not an appropriate quantity for robust statistical inference, as if the model has a variance parameter, the fitted model $P_{\hat{\theta}_{\lambda}}$ can collapse onto a data point. Lemmas \ref{prop:smallLambda} and \ref{largeLambdaEmpiricalBehavior} thus indicates that the tiny $\lambda$ and massive $\lambda$ regimes are not appropriate for simulteneously achieving both informative and robust inferences. The following lemma indicates that the $\lambda$-RSW Divergence achieves theoretical robustness in its value when $\lambda = 1$, in the sense that if $P = (1-\epsilon)Q+\epsilon F$ where $W_2(Q,P_{\theta^{*}}) \leq \rho$ and $\epsilon$ and $\rho$ are small, then the value of loss at $\theta^{*}$ is within $\epsilon+\epsilon^{2}+\rho^2$ of the value of the loss at the infima.


\begin{restatable}{lemma}{RobustnessUnderContaminationModel}(Robustness to G+H Contamination)
    \label{lem:mathRobust}
    Let $\{P_{\theta} : \theta \in \Theta\}$ satisfy $P_{\theta} \in \mathcal{P}_{2}(\mathbb{R}^{m})$. Suppose for some $\theta^{*} \in \Theta$, $Q,F \in \mathcal{P}_2(\mathbb{R}^{m})$, $0 < \epsilon < 1-\frac{1}{\sqrt{2}}$ and $\rho >0$, $P = (1-\epsilon) Q + \epsilon F$
    where $W_2(Q,P_{\theta^{*}}) \leq \rho$. Then for $\lambda >0$,
        $\ell_{\lambda}(P,P_{\theta^{*}}) \leq \inf_{\theta \in \Theta} \ell_{\lambda}(P,P_{\theta}) + \frac{\epsilon+\epsilon^2}{\lambda}+\rho^2$.
    \text{Also, for any }$\eta,\zeta>0$, $\theta_{\lambda,\eta}, \theta^{*} \in \Omega$ where
    \[\Omega := \{\theta \in \Theta: \exists Q, \mathrm{KL}(Q,P)+\lambda W_2^2(Q,P_{\theta}) \leq \epsilon+\epsilon^2+\lambda (\rho^2 + \eta+\zeta)\}.\]
    \label{lem:robustnessOfDivergence}
\end{restatable}
Lemma \ref{lem:mathRobust} is a mathematical justification for the robustness of the minimum $\ell_{\lambda}$ parameter $\theta_{\lambda,\eta}$ because
for intermediate values of $\lambda$ (e.g $\lambda = 1$) and when $\rho$ and $\epsilon$ are small (and $\eta$ is chosen arbitrarily small), $\Omega$ contains only those parameters $\theta$ such that $P_{\theta}$ is geometrically close to a third distribution $Q$ which is a small $\mathrm{KL}$ re-weighting away from $P$. $\theta^{*}$ is one such parameter; there may be others, and $\theta_{\lambda,\eta}$ will be one of them (as $\eta$ is small). Unlike the Huber contamination model, the GH contamination model is inherently non identifiable even if the statistical model is identifiable in the sense that there can be infinitely many $\theta^{*}$ satisfying $P = (1-\epsilon)Q+\epsilon F$ such that $W_2(Q,P_{\theta^{*}}) \leq \rho$. This is why the best we can hope for (without imposing additional structure to identify $\theta^{*}$) is to recover a $\theta^{*} \in \Theta$ such that $P_{\theta^{*}}$ is a small geometric perturbation from a third distribution which is a small $\mathrm{KL}$ measured reweighting away from $P$. We emphasize that selecting an intermediate value for $\lambda$ is key to achieving robust inferences. $\lambda =1$ is a good starting point (since then $\Omega$ consists only of $\theta$ such that for some third distribution $Q$, $\max(\mathrm{KL}(Q,P),W_2^2(Q,P_{\theta})) \leq \epsilon+\epsilon^{2}+\rho^2+\eta+\zeta $) (for arbitrarily small $\eta,\zeta$) but in Section \ref{sec:lamSelect} we provide a systematic data-driven approach for $\lambda$ selection.

\section{Algorithm for robust SBI}
\label{sec:algorithms}

By Definition \ref{def:RSW2},  $\ell_{\lambda}(P_n,P_{\theta})$ is itself defined through a complex min-min problem. Finding $\hat{\theta}_{\lambda}$ thus appears to be a difficult min-min-min problem. 
In Section \ref{subsec:minmax}, we introduce a reformulation of ${\ell}_{\lambda}$ as a concave maximization problem in the semi-discrete case, and provide a stochastic sub-gradient ascent algorithm to solve it. In Section \ref{subsec:FullSBIpipeline} we provide the pipeline for SBI with uncertainty quantification using minimum $\lambda$-RSW estimation for a fixed $\lambda$. As $\lambda$ selection is critical to robustness, we describe in Section \ref{sec:lamSelect} a method for $\lambda$ selection that can be used before executing the algorithm of Section \ref{subsec:FullSBIpipeline}. 

\subsection{From min-min-min to min-max}
\label{subsec:minmax}

\begin{restatable}{theorem}{minMaxFlip}
    \label{thm:minMaxFlip}
    If $P_n$ is an empirical measure at points $Y_1,Y_2,\dots,Y_n \in \mathbb{R}^{m}$, and $F \in \mathcal{P}_2(\mathbb{R}^{m})$ is absolutely continuous with respect to Lebesgue measure, then
    \begin{equation}
        \label{dualExpectation}
        \begin{split}
        \ell_{\lambda}(P_n,F) = \sup_{g \in \mathbb{R}^{n}} \mathbb{E}_{X \sim F} h_1(X,g)
        \end{split}
    \end{equation}
     where $h_1: \mathbb{R}^{m} \times \mathbb{R}^{n} \to \mathbb{R}$ satisfies
     \vspace{-2pt}
    \begin{equation}
        \label{h1Def}
    h_1(x,g) =   \min_{j \in [n]} (\| x - Y_j \|_2^2 - g_j) - \frac{1}{\lambda} \log (\sum_{t=1}^{n} \frac{1}{n} \exp(-\lambda g_t)).
    \end{equation}
    \vspace{-2pt}
    and the suprema is achieved. Additionally if $\hat{Q}_{\lambda}$ is a discrete measure supported at the points $Y_1,Y_2,\dots,Y_n$ with weights constructed by choosing an arbitrary argument maximizer $g^{*}$ of Eqn. \ref{dualExpectation} and setting the $j^{th}$ weight to $\frac{\exp(-\lambda g_j^{*})}{\sum_{t=1}^{n} \exp(-\lambda g_t^{*})}$ for $j \in [n]$, then
    \vspace{-2pt}
    \begin{equation}
    \label{optWeightVectorExtraction}
    \hat{Q}_{\lambda} = \argmin_{Q \ll P_n} \frac{1}{\lambda} KL(Q,P_n) +  W_2^2(F,Q).
    \end{equation}
    \vspace{-2pt}
\end{restatable}
Since $h_1$ is concave in $g$ (see Lemma \ref{lem:infoAboutOurObjective}), computation of $\ell_{\lambda}(P_n,F)$ can be approximated using a stochastic sub-gradient ascent. Our approach, presented in Algorithm \ref{alg:SGA}, uses a carefully chosen learning rate to achieve the convergence rate detailed in Section \ref{sec:convergenceProperties}. The learning rate scale input parameter $B$ will be discussed in Section \ref{sec:convergenceProperties}.

 \begin{algorithm}
\caption{$\lambda$-RSW Divergence Approximation via SGA}
\begin{algorithmic}[1]   
\Require $g_{0} \in \mathbb{R}^{n}$, number of iterations $s$, $\lambda > 0$, $B >0$ (learning rate scaling), $\mathcal{Y} = \{Y_1,Y_2,\dots,Y_n\}$, $\mathcal{X} = \{X_1,\dots,X_s\}$ where $X_1,\dots,X_s \overset{iid}{\sim} P_{\theta}$.
\Ensure Estimate of $\ell_{\lambda}(P_n,P_{\theta})$ where $P_n$ is empirical measure of $\mathcal{Y}$, and final iterate $g_{s}$
\State $RunningAvg = 0$
\For{$i = 1,\dots,s$}
    \State $\gamma_i = B\sqrt{\frac{n}{i}}$ (set learning rate)
    \State $RunningAvg = RunningAvg+\gamma_i  h_1(X_i,g_{i-1})$ (update running average)
    \State Compute $j^{*} = \argmin_{j \in [n]} \| X_i - Y_j \|_2^2 - g_{i-1,j}$ and $\sum_{\ell = 1}^{n} \exp(-\lambda g_{i-1,\ell})$
    \For{$j = 1,\dots,n$}  
        \State $g_{i,j} = g_{i-1,j} + \gamma_i \left( \frac{\exp(-\lambda g_{i-1,j})}{\sum_{\ell=1}^{n} \exp(-\lambda g_{i-1,\ell})} - \mathbb{I}(j = j^{*}) \right)$ (update sub-grad)
    \EndFor
\EndFor
\State \Return $(\frac{1}{\sum_{i=1}^{s} \gamma_i} RunningAvg,g_s)$
\end{algorithmic}
\label{alg:SGA}
\end{algorithm}

\subsection{Full SBI algorithm}
\label{subsec:FullSBIpipeline}

Theorem \ref{thm:minMaxFlip} also makes clear that the minimum $\lambda$-RSW problem is a non-convex concave stochastic min-max problem.  In particular for a statistical model $\{P_{\theta}: \theta \in \Theta\}$ where each member $P_{\theta} \in \mathcal{P}_{2}(\mathbb{R}^{m})$ is absolutely continuous with respect to Lebesgue measure,$
\inf_{\theta \in \Theta} \ell_{\lambda}(P_n,P_{\theta}) = \inf_{\theta \in \Theta} \sup_{g \in \mathbb{R}^{n}} \mathbb{E}_{X \sim P_{\theta}}  h_1(X,g).$
Algorithms for solving non-convex concave stochastic min-max problems via alternating stochastic gradient ascent-descent can at best guarantee convergence to a stationary point in the $\Theta$ space \citep{lin2020gradient}, and so we instead use a global function optimizer that only requires the ability to evaluate the function at any $\theta$, which we can do by using Algorithm \ref{alg:SGA} to approximate  $\ell_{\lambda}(P_n,P_{\theta})$. Our choice is the Covariance Matrix Adaptive-Evolution Strategy algorithm (CMA-ES) and we use the implementation provided in the \verb|cma| package in Python \citep{hansen2006cma}. The algorithm takes as input a population size $K$, number of rounds $R$, initial step size $\sigma_0 >0$, lower and upper bounds on the support of each parameter if they exist, and initial position $\theta^{0} \in \Theta$ where $\Theta \subset \mathbb{R}^{d}$ for some $d \geq 1$. We let $\Omega$ denote the full set of optimization tuning parameters. At round $i$, $K$ values in $\Theta$ are randomly sampled according to $Normal(\theta^{i},\sigma_i^2 C_{i})$ where $C_0 = \mathbf{I}_{d}$.  Then the function value is computed at each of the samples, and these are used to produce a new estimate for the minimum of the function, $\theta^{i+1}$, an updated step-size $\sigma_{i+1}$ and an updated covariance matrix around the estimate, $C_{i+1}$. CMA-ES comes with an automated strategy for bringing samples of $\theta$ back into $\Theta$ if they are sampled outside the bounds. An advantage of CMA-ES is that it can leverage parallel computing resources at each round, when the $R$ function values are computed.

Algorithm \ref{alg:singleBootstrapSample} is repeated $M$ times to produce the bootstrap samples $\{\theta_{1}^{boot},\dots,\theta_{M}^{boot}\}$ for uncertainty quantification for the parameter. We call this procedure B-MRSW, standing for (B)ootstrapped (M)inimum (R)obust (S)emi-constrained (W)asserstein-2. $M \times K \times R \times s$ executions of the simulator $G$ are used to generate $M$ bootstrap samples. Note that prior to executing Algorithm \ref{alg:singleBootstrapSample}, $s$ noises $Z_1,\dots,Z_s \overset{iid}{\sim} \mu$ must be generated where $\mu$ is the reference distribution.

\begin{algorithm}
\caption{Generation of a Bootstrap sample}
\label{alg:singleBootstrapSample}
\begin{algorithmic}[1]
    \Require $\lambda > 0$, observed data $\{Y_i\}_{i=1}^n$, simulator $G$, $g_0$ (initial position for SGA), $s$ iterations per SGA, $B$ (learning rate scaling in SGA), $s$ noise vectors $\mathcal{Z} = \{Z_1,\dots,Z_s\}$ (where $Z_1,\dots,Z_s \overset{iid}{\sim} \mu$), $\Omega = \{\text{CMA-ES tuning parameters (see Section \ref{subsec:FullSBIpipeline})}\}$
    \State Sample uniformly with replacement $n$ times from $(Y_1, \dots, Y_n)$ to generate a bootstrapped dataset $\mathcal{Y}^{'}= (Y_1', \dots, Y_n')$ with empirical measure $P'_n$
    \State Find the $j^{th}$ bootstrap sample $\theta^{boot} = \argmin_{\theta \in \Theta} \ell_{\lambda}(P'_n, P_{\theta})$ using CMA-ES with params $\Omega$
    \State Note: For each function evaluation at $\theta \in \Theta$ in CMA-ES in line (2), algorithm \ref{alg:SGA} executes with $g_0,s,B,\lambda,\mathcal{Y}^{'}$ and $\mathcal{X} = (X_1,\dots,X_s)$ where $X_i = G(Z_i,\theta)$ for $i \in [s]$
    \State \Return $\theta^{boot}$
\end{algorithmic}
\end{algorithm}

\subsection{$\lambda$ selection routine}
\label{sec:lamSelect}
$\lambda$ is the lever for data reweighting, and as detailed in Section \ref{sec:basicProperties}, if $\lambda$ is too close to zero, inferences lack robustness to Huber contamination because Wasserstein-2 distance is not robust to this form of contamination, whereas if $\lambda$ is too large, the loss landscape has lost the richness of Wasserstein distance and inference is uninformative. To achieve robustness and maintain a rich loss landscape requires appropriate selection of $\lambda$. To construct a data driven strategy for $\lambda$ selection, note that by Theorem \ref{thm:minMaxFlip} Equation \ref{optWeightVectorExtraction} the re-weighting of the data $Y_1,\dots,Y_n$ achieving the infima in $\ell_{\lambda}(P_n,P_{\theta})$ can be estimated by extracting the final iterate $g_s(\theta)$ from Algorithm \ref{alg:SGA} (when applied with given $\theta$) and computing $\hat{Q}_{\lambda}(\theta) := \sum_{t=1}^{n} w_t(\theta) \delta_{Y_t}$ where  $w_t(\theta) = \frac{\exp(-\lambda g_{s,t}(\theta))}{\sum_{j=1}^{n} \exp(-\lambda g_{s,j}(\theta))}$ for $t \in [n]$. 

Our approach is to leverage this by carrying out a small preliminary $\lambda$ selection routine before executing the main bootstrap loop described in Section \ref{subsec:FullSBIpipeline}. 
We take a collection of logarithmically spaced $\lambda$ values $\Lambda$ where for $M'$ bootstrap samples, and $j \in [M']$ we first draw the noises $Z_1,\dots,Z_s \overset{iid}{\sim} \mu$. Then for each $\lambda \in \Lambda$, we run Algorithm \ref{alg:singleBootstrapSample}, and then approximate $W_2(P_{\theta^{boot,j}_{\lambda}}, \hat{Q}_{\lambda}(\theta^{boot,j}_{\lambda}) )$ (using the existing samples $X_1,\dots,X_s \overset{iid}{\sim} P_{\theta_{\lambda}^{boot,j}}$ to perform a standard discrete-discrete Optimal Transport computation). At each $\lambda$ this produces a set of $j$ estimates of the Wasserstein-2 distance between the fitted model and the reweighting of the data corresponding to this fitted model. 

We then look for an elbow in the plot of the distribution of $W_2(P_{\theta^{boot,j}_{\lambda}}, \hat{Q}_{\lambda}(\theta^{boot,j}_{\lambda}) )$ across $\lambda \in \Lambda$. The intuition is that as $\lambda$ grows, outliers corresponding to the Huber contaminant become down-weighted more in $\hat{Q}_{\lambda}(\hat{\theta}_{\lambda})$. Once the outliers are nearly completely downweighted (i.e $\lambda$ is high enough), the quantity $W_2(P_{\hat{\theta}_{\lambda}},\hat{Q}_{\lambda}(\hat{\theta}_{\lambda}))$ will stop decreasing sharply, yielding an elbow or U shape. If no elbow is discovered, $\lambda = 0$ is chosen. We find that choosing $\Lambda$ to be $15$ logarithmically equidistant values between $-2$ and $2$ works well.  We demonstrate the $\lambda$ selection method for a univariate normal distribution example in Figure \ref{fig:lamSelect} (code for all experiments is available at \verb|https://github.com/pjacobs0007-lang/BMRSW|). We provide  additional demonstrations of $\lambda$ selection across different contamination settings in Appendix Section \ref{sec:experimentalDetails}.

\begin{figure}
  \centering
  \includegraphics[width=\textwidth]{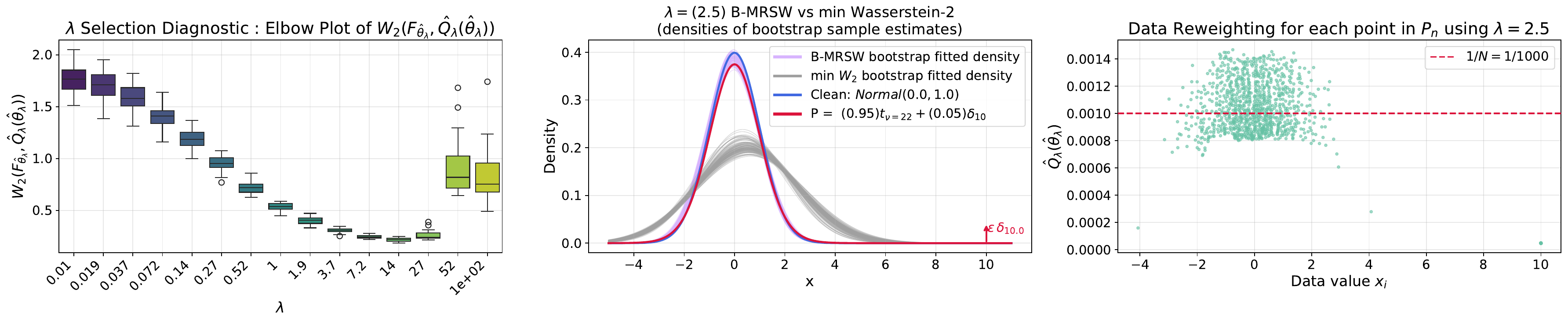}
  \caption{Illustration of the full B-MRSW inference pipeline and comparison to Minimum $W_2$ in inferring the mean and variance of a univariate normal distribution. Here $P = (.95) t_{\nu = 22}+(.05)\delta_{10}$ where $t_{\nu}$ denotes the $t$ distribution with $\nu$ degrees of freedom where $n = 1000$.
  Left: Elbow Plot of $\lambda$ selection diagnostic (Elbow observed at $\lambda \approx 2.5$). Middle: Comparison of fitted densities from $100$ bootstrap samples of the $\lambda = (2.5)$ B-MRSW vs $100$ bootstrap samples of Minimum $W_2$. Right: The reweighted empirical measure $P_n$ determined by running minimum $\lambda = (2.5)$ RSW directly on $P_n$ and then computing $\hat{Q}_{\lambda}(\hat{\theta}_{\lambda})$.}
  \label{fig:lamSelect}
\end{figure}
\section{Convergence guarantee for divergence approximation using SGA}
\label{sec:convergenceProperties}

Here we prove a rate of convergence for Algorithm \ref{alg:SGA} in approximating $\ell_{\lambda}(P_n,P_{\theta})$.

\begin{restatable}{theorem}{sgaConvergence}
    \label{thm:sgaConvergence}
    Let $Y_1,\dots,Y_n \in A \subset \mathbb{R}^{m}$, $P \in \mathcal{P}_2(A)$, $\sup_{x,y \in A} \| x- y\|_2 \leq D$, $\mathbf{0}_{m} \in A$ and $\lambda > 0$. Then letting $\hat{x}_{s,g_0,B}$ be the output of algorithm \ref{alg:SGA} with $s$ iterations, learning rate scaling $B$ and $g_{0} = \pmb{0}_n$,
    \begin{equation}
        \label{sgaPrecision}
        \mathbb{E}_{P} | \hat{x}_{s,g_0,B} -  \ell_{\lambda}(P_n,P) | \leq 3 D^2 \sqrt{\frac{1+\log(s)}{s}} + \sqrt{\frac{n}{s}} \begin{cases}
             \left( \frac{D^2}{2} + 2 D^2(1+\log(s))\right)& B = D^2 \\
            \frac{D^{4}}{2} + 2(1+\log(s))& B = 1
        \end{cases}
    \end{equation}
    Moreover, $\hat{x}_{s,g_0,B}$ is computed in $O(ns)$ time.
\end{restatable}
The purpose of including the learning rate scaling $B$ in the algorithm is because if one has access to $D$, they can use it to reduce the diameter penalty in the rate from $D^{4}$ to $D^{2}$, which is seen Theorem \ref{thm:sgaConvergence}. The following corollary summarizes the asymptotics of Theorem \ref{thm:sgaConvergence}.
\begin{corollary}
    Under the conditions of Theorem \ref{thm:sgaConvergence} and absorbing $D$ into asymptotic notation, for $\epsilon > 0$, $s = n \epsilon^{-2}$ and $B \in \{D^2,1\}$,
    $\mathbb{E}_{P} | \hat{x}_{s,g_0,B} - \ell_{\lambda}(P_n,P) | = \tilde{O}(\epsilon) \text{ with runtime } O(\frac{n^2}{\epsilon^2})$.
\end{corollary}
where $\tilde{O}$ indicates ignorance of logarithmic factors of $\epsilon$. For estimation of $\ell_{\lambda}(P,Q)$ where $P$ and $Q$ are both discrete, \citet{le2021robust} provide an algorithm which deterministically achieves $\epsilon$ precision in absolute error, in $O(\frac{n^2}{\epsilon})$ time. The semi-discrete problem is naturally more challenging because one of the measures is potentially smooth and unknown. To the best of our knowledge, Theorem \ref{thm:sgaConvergence} is the first algorithm with a convergence guarantee for the semi-discrete Robust Semi-constrained Optimal Transport.

\section{Numerical illustration and comparison with MMD}
\label{sec:numericalIllustrations}

The G-and-K distribution \citep{prangle2017gk} is a flexible four parameter $(a,b,g,k)$ extension of the univariate normal distribution incorporating two additional parameters which respectively measure skewness ($g$) and kurtosis ($k$). A random variable is generated from a four parameter G-and-K distribution with parameters $(a,b,g,k)$ by first generating $Z \sim Normal(0,1)$, and then passing it through the transformation $Y = a+b G_{g}(Z) K_{k}(Z)$
where $G_{g}(z) = 1+(0.8) \tanh( gz /2)$ and $K_{k}(z) = z (1+z^2)^{k}$. $a$ primarily controls location and $b$ primarily variance. The density of the G-and-K distribution is not analytically tractable and so it is a common benchmark for SBI.  For the simulation study the contaminated data generating distribution $P = (1-\epsilon) Q + \epsilon \delta_{z}$ where $Y \sim Q$ if $X \sim GandK(a = 3, b = 1,g = 2, k =0.5)$ and $Y = \lfloor \frac{X}{\rho} \rfloor \rho$. We choose $\epsilon = \rho = 0.05$ and $z =50$. $P_{clean}$ and the contaminated data generating distribution are visualized in figure \ref{fig:gandkCleanvsContam} Panel A. We compare to NPL-MMD of \cite{dellaporta2022robust}. Although MMD has not been studied under geometric contamination, several sources \citep{cherief2020mmd,dellaporta2022robust} have suggested that MMD is robust under Huber Contamination. Given a positive definite kernel function, $k:\mathbb{R}^{m} \times \mathbb{R}^{m} \to \mathbb{R}^{+}$, the MMD between two Borel probability measures $P$ and $Q$ on $\mathbb{R}^{m}$ can be expressed as
\begin{equation}
    \label{mmdDefinition}
    MMD^2_{k}(P,Q) = \mathbb{E}_{X \sim P, Y \sim P} k(X,Y) + \mathbb{E}_{X \sim Q, Y \sim Q} k(X,Y) - 2 \mathbb{E}_{X \sim P, Y \sim Q} k(X,Y).
\end{equation}
\cite{dellaporta2022robust} empirically demonstrate that by using the popular Gaussian kernel, $k(x,y) = \exp(\frac{-(x-y)^2}{2 \sigma_0^2})$, their SBI mechanism of using the Bayesian bootstrap on top of Minimum squared MMD achieves robustness to Huber contamination.\footnote{NPL-MMD also has the flexibility to incorporate prior knowledge, but here we consider the case prior knowledge is not available.} They choose $\sigma_0 = .15$ without providing a systematic procedure for selection. For a dataset of $N = 1000$ points from the contaminated distribution, we carry out the $\lambda$ selection procedure of Section \ref{sec:lamSelect} with $M' = 15$ bootstrap replications and $s = 20,000$ noise generations per bootstrap replication (with the conservative learning rate scaling $B = 1$) for $\Lambda = \{10^{-2+\frac{4k}{14}}: k \in \{0,1,\dots,14\} \}$ (see appendix Section \ref{sec:experimentalDetails} for our CMA-ES optimization parameter settings). The results of applying the $\lambda$ selection procedure are in Figure \ref{fig:gandkCleanvsContam} panel B. A conservatively large regional estimate for the elbow is $\lambda \in [0.5,3.5]$. To compare the sensitivity of NPL-MMD inferences to bandwidth selection, compared to the sensitivity of B-MRSW to $\lambda$ selection, and to assess how both methods handle G+H Contamination, we then generate $20$ datasets for each $n \in \{1000,5000\}$, and for each run B-MRSW and NPL-MMD with $M = 100$ bootstrap samples, with $\lambda \in \{0.5,1.5,2.5,3.5\}$ for B-MRSW and $\sigma_0 \in \{0.15,0.30,0.50,0.65,1\}$ for NPL-MMD. For each dataset, for both methods, we compute the marginal median estimate for each parameter based on the $M$ bootstrap samples. We also compute marginal $95\%$ confidence intervals for each parameter. To calibrate the comparison with NPL-MMD we use a common budget for noise generation for both methods. See Appendix Sections \ref{subsec:gandkExperimentDetails} and \ref{subsec:simEfficiencyComparison} for complete NPL-MMD Stochastic Gradient Descent settings and a comparison of simulation efficiency of the two methods.


Panel C and D of Figure \ref{fig:gandkCleanvsContam} respectively show the average MSEs of the marginal median estimators for both methods. We observe that small increases to the bandwidth lead to highly inaccurate estimates of the skewness $g$ for NPL-MMD, whereas for B-MRSW, average MSEs for all 4 parameters across all $\lambda$ values still provide accurate inferences. We note that the small bandwidth limit in Gaussian Kernel MMD suffers from a loss of geometric informativeness (which manifests as a vanishing gradient problem) because the distance mass is shifted in space is no longer properly encoded, whereas large bandwidth MMD with the Gaussian kernel is not inherently robust (as bandwidth tends to infinity Gaussian Kernel MMD tends towards a distance between first moments; see Appendix Section \ref{sec:MMDnotRobust} for details). Panels E and F of Figure \ref{fig:gandkCleanvsContam} present the results of the coverage experiment. Coverage rates for NPL-MMD are brittle for $g$ as a function of bandwidth. For B-MRSW, coverage rates are high across $\lambda$ at $n = 5000$ except for kurtosis, for which B-MRSW coverage improves considerably as $n$ increases. 


\begin{figure}[h!]
    \includegraphics[width=\textwidth]{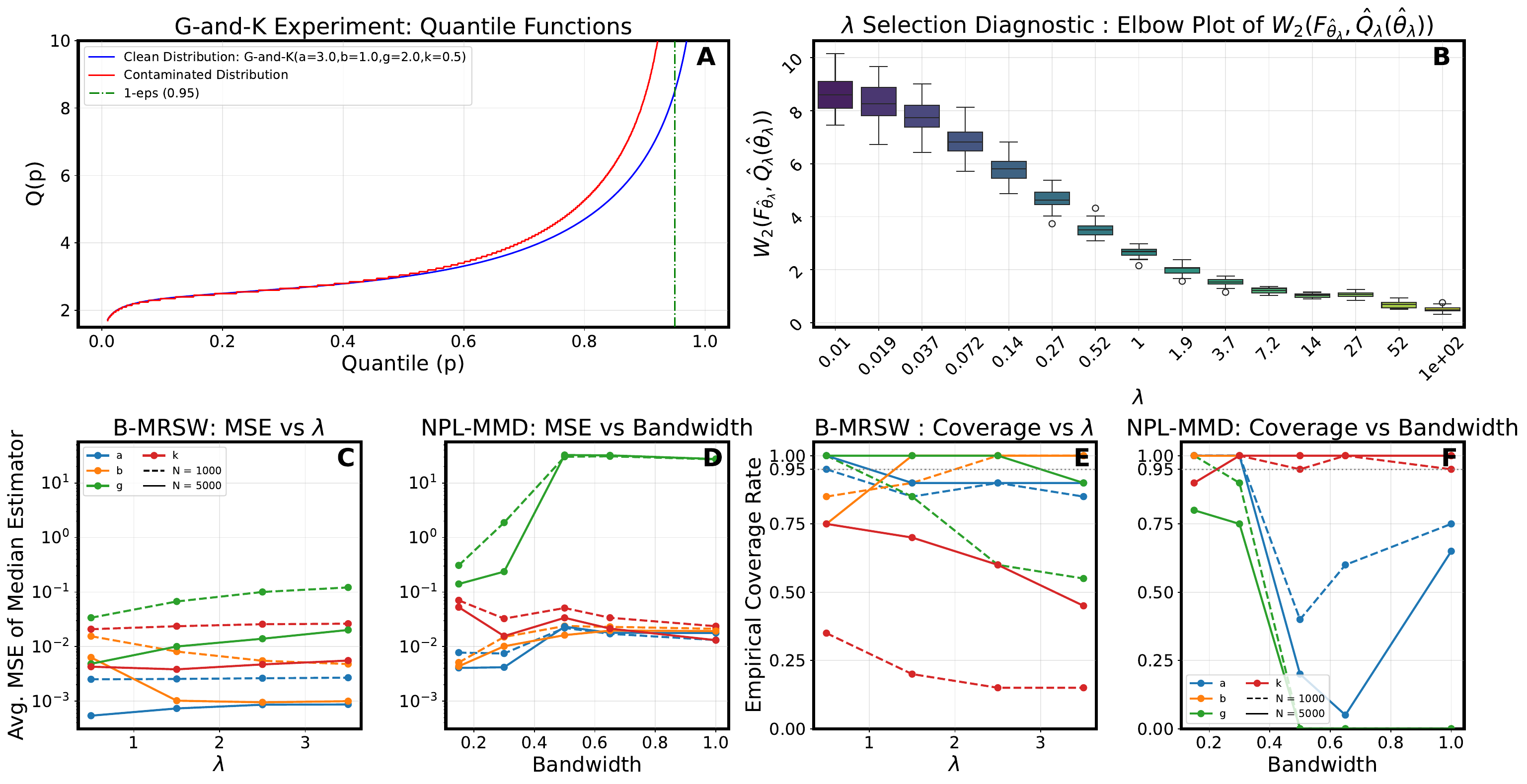}
    \caption{ Panel A: Visualization of the quantile functions of the clean distribution (G-and-K $(a=3,b=1,g=2,k=0.5)$ and Contaminated distribution $P$. Panel B: Results of applying the $\lambda$ selection diagnostic routine. Panel C and D: Average MSE of Bootstrap marginal Median estimators for B-MRSW and NPL-MMD. Panels E and F: Coverage rates for B-MRSW and NPL-MMD.} 
\label{fig:gandkCleanvsContam}
\end{figure}
\section{Conclusions}
We have presented B-MRSW, a framework for statistical inference for a parameter with uncertainty quantification. It works in the SBI regime, is robust to a general form of contamination involving both a geometric perturbation and a mass contamination, and is parallelizable, both at the bootstrap level, and within each individual optimization. Relative to existing works which can be robust to the less general Huber Contamination, B-MRSW has a principled approach for tuning parameter selection, a critical step in the robust inference pipeline. A limitation of B-MRSW is that it executes the simulator code $G$ a total of $(M+M'|\Lambda|$)*\verb|CMA-ES pop_size|*\verb|CMA-ES rounds|*\verb|iterations_per_SGA| where $M'$ is bootstrap replications during $\lambda$ selection and $M$ is bootstrap replications during inference. The total calls to the simulator can be significantly reduced by training a non-parametric model for the simulation code $G$ and plugging it directly into the existing framework, a subject for future work.
\bibliographystyle{plainnat}
\bibliography{references}

\appendix

\section{Empirical likelihood motivation for robust semi-constrained optimal transport}
\label{sec:EmpiricalLikelihoodMotivation}

Our approach for divergence selection is motivated by recent developments in distributionally constrained Empirical Likelihood.  \cite{chakraborty2025robust} develop a robust Bayesian method that centers around a parametric family $\{P_\theta \,:\, \theta \in \Theta\}$ while allowing flexible departures from it in a holistic manner. Their main idea was to {\it replace} the usual likelihood function $L(\theta) = \prod_{i=1}^{n} P_{\theta}(Y_i)$ with a principled notion of {\it distributionally-constrained empirical likelihood} $L_{\rm DCEL}$, 
\begin{align}\label{eqn:wbetel}
L_{\rm DCEL}(\theta) :\,= \prod_{i=1}^n w_i^\star(\theta), \ \text{ where } \ w^\star(\theta) = \argmax_{w \in \Delta^{n-1}} \prod_{i=1}^n w_{i}^{-w_{i}} \ \text{ subject to } \ \mbox{D}[Q_{n,w},P_{\theta}] \le  \varepsilon,
\end{align}
where $\varepsilon > 0$ is a pre-specified tolerance level, $D$  an appropriate statistical distance, 
and $Q_{n,w} :\, = \sum_{i=1}^n w_i \delta_{Y_i}$ a {\it weighted empirical distribution} with weight vector $w \in \Delta^{n-1}$, the $(n-1)$-dimensional probability simplex ($\Delta^{n-1} = \{(w_1,w_2,\dots,w_n) \in \mathbb{R}^{n} : \forall i,  0 \leq w_i , \sum_{i=1}^{n} w_i = 1\})$. The idea of constructing a likelihood function based on a {\it discrete distribution} supported on the {\it observed data points} falls under the umbrella of {\bf empirical likelihood} (EL; \citep{owen2001empirical,lazar2021review}). Observe that $\log \prod_{i=1}^n w_i^{-w_i} = -\sum_{i=1}^n w_i \log w_i$ is the {\it Shannon entropy} \citep{cover1999elements} of $Q_{n,w}$. In \eqref{eqn:wbetel}, the optimal weight vector $w^\star(\theta)$ underlying $L_{\rm DCEL}(\theta)$ is obtained by maximizing the Shannon entropy of $Q_{n,w}$ subject to $Q_{n,w}$ being within an $\varepsilon$-distance of $P_\theta$ w.r.t. the divergence $\mbox{D}$. The maximum entropic notion was specifically motivated by {\it exponentially tilted empirical likelihood} (ETEL; \citep{schennach2005bayesian,chib2018bayesian}),  traditionally used when the parameter $\theta$ is defined {\it implicitly} via {\it moment conditions}, e.g. in econometric applications. 

Rather than taking the likelihood based approach as in \cite{chakraborty2025robust}, in this work we construct a statistical divergence driven by the soft-dual formulation of the optimization problem \ref{eqn:wbetel}, and implement it with the geometrically sensitive choice of $\mbox{D} = W_2^2$. Specifically, observe that with $P_n := \sum_{i=1}^{n} \frac{1}{n} \delta_{Y_i}$ denoting the empirical measure of the observed data $Y_1,\dots,Y_n$,
\[
w^{*}(\theta) = \argmin_{w \in \Delta^{n-1}} \mathrm{KL}(Q_{n,w},P_n) \text{ subject to } W_2^2[Q_{n,w},P_{\theta}] \leq \epsilon.
\]
The soft-dual to this optimization problem finds for a specified $\lambda > 0$,
\begin{equation}
    \label{eqn:softDualExpression}
    \argmin_{w \in \Delta^{n-1}} \mathrm{KL}(Q_{n,w},P_n) + \lambda W_2^2(Q_{n,w},P_{\theta}).
\end{equation}

We take the quantity $\min_{w \in \Delta^{n-1}} \frac{1}{\lambda} \mathrm{KL}(Q_{n,w},P_n)+ W_2^2(Q_{n,w},P_{\theta})$ to be an empirical estimate of the divergence between $P_{\theta}$ and $P$; it is exactly $\ell_{\lambda}(P_n,P_{\theta})$, the quantity that we use as our divergence in this work. This is an entirely different motivation behind the divergence than the one provided by \cite{le2021robust}, where it arises from Unbalanced Optimal Transport \citep{chizat2017unbalanced}.

\section{Extended related work}
\label{sec:ExtendedRelatedWork}

Robust parametric statistics deals with the general setting where the parameters of a statistical model $\{P_{\theta}: \theta \in \Theta\}$ are to be inferred from data generated according to a probability distribution $P \notin \{P_{\theta}: \theta \in \Theta\}$, but $P$ is in some sense $\epsilon$ perturbed from $P_{\theta^{*}}$. A statistical method is robust if the small $\epsilon$ perturbation to the data generating distribution cannot be made to cause massive changes in inferences. What is meant by a small perturbation is made explicit using a contamination model. \citet{huber1992robust} influentially introduced the Huber Contamination model, where $P = (1-\epsilon) P_{\theta^{*}}+\epsilon F$ for some probability distribution $F$. Alternatively, the Adaptive Contamination model, also sometimes called the strong $\epsilon$-contamination model, considers the case where an adversary can take an $\epsilon$ fraction of the samples from the empirical measure of $n$ iid samples from $P_{\theta^{*}}$, and arbitrarily modifies them \citep{diakonikolas2019robust}. \cite{liu2023robust} studies the adversarial Wasserstein contamination to the empirical measure, where the adversary can apply a slight perturbation to every data point. At the population level, the adaptive contamination model and adversarial Wasserstein contamination model correspond to $TV(P_{\theta^{*}},P) \leq \epsilon$ and $W_1(P_{\theta^{*}},P) \leq \epsilon$ respectively. We are not aware of works that point out the utility of a method for simultaneously addressing a geometric and Huber contamination.

Although our own method is not Bayesian, our method has the same inferential goal (to provide uncertainty quantification for the parameter of the statistical model), hence here we provide a summary of SBI approaches that includes Bayesian approaches. In traditional Bayesian inference, the goal is  to target the posterior distribution $p(\theta | X_1,\dots,X_n) \propto \prod_{i=1}^{n} p_{\theta}(X_i) \pi(\theta)$ of $\theta$ where $p_{\theta}$ is the probability density function associate to $P_{\theta}$ and $\pi$ is a prior distribution. To handle the challenge that $p_{\theta}$ is not analytically available, there are various existing approaches. In Approximate Bayesian Computation \citep{beaumont2010approximate,beaumont2002approximate,blum2010approximate,beaumont2002approximate,lintusaari2017fundamentals}, samples from an $\epsilon$-indexed Posterior distribution are generated by sampling $\theta_j \sim \pi$ from the prior, generating a dataset $X_1',X_2',\dots,X_n' \overset{iid}{\sim} p_{\theta_j}$, and then accepting $\theta_j$ as sample from the Posterior if $D(P_n,P_j) \leq \epsilon$ where $P_j$ is the empirical measure of $(X_1',\dots,X_n')$ and $D$ is a carefully chosen measure of discrepancy. Early methods were based on choosing $D$ to measure distance between summary statistics. More recent attempts \citep{bernton2019approximate} use non-parametric losses such as Wasserstein distance. When $\epsilon$ is small, the simulation cost can be high, which has motivated alternative approaches. Neural Likelihood Estimation \citep{papamakarios2019sequential} also uses the ability to simulate pairs $(\theta,x)$ such that $x | \theta \sim p_{\theta}$ and $\theta \sim \pi$. But they use them to train a Neural Network model $\hat{q}(x | \theta)$ which approximates $p_{\theta}(x)$. The model of the likelihood can then be used in a frequentist inference framework, or plugged into Markov Chain Monte Carlo for Bayesian Inference \citep{brooks1998markov}. Neural Posterior Estimation \citep{papamakarios2016fast,greenberg2019automatic} uses a Neural network to directly learn the posterior distribution $p(\theta | x)$.

Likelihood based inference and subsequently Bayesian inference are generally not robust to mis-specification of the likelihood; specific to Bayesian inference, posterior credible intervals can fail to have the correct coverage \citep{kleijn2012bernstein}. To address these concerns, \cite{miller2019robust} introduce the Coarsened Posterior, which changes the conditioning event in Bayesian inference from $(X_1,\dots,X_n) = (x_1,\dots,x_n)$ to $D(\sum_{j=1}^{n} \frac{1}{n}\delta_{X_j}, \sum_{j=1}^{n} \frac{1}{n} \delta_{x_j}) \leq \epsilon$ for some choice of $D$ and $\epsilon$. For $D = \mathrm{KL}$, they show that for a $\zeta_n > 0$, the coarsened posterior is approximately proportional to $\exp(\zeta_n  \sum_{j=1}^{n} \log(p_{\theta}(X_j)) ) \pi(\theta)$, which is called the Power or Fractional posterior \citep{bhattacharya2019bayesian}. As likelihood based inference is still not robust to Huber contamination (maximum likelihood targets $\min_{\theta \in \Theta} \mathrm{KL}(P,P_{\theta})$), a popular extension used in SBI for robustifying Bayesian inference is the generalized posterior $p(\theta | X_1,\dots,X_n) \propto \exp(-\beta n \mathcal{L}(P_{\theta},P_n)) \pi(\theta)$ \citep{bissiri2016general,lyddon2019general}, where $\mathcal{L}$ is a robust empirical loss and $\beta$ is a user specified parameter. $\beta =1$ and $\mathcal{L}(P_n,P_{\theta}) = -\frac{1}{n} \sum_{i=1}^{n} \log p_{\theta}(X_i)$ recovers the original, non-robust posterior distribution (which targets $\mathrm{KL}(P,P_{\theta})$). MMD-Bayes, introduced by \cite{cherief2020mmd}, uses an unbiased estimate of squared Maximum Mean Discrepancy (MMD) as $\mathcal{L}(P_{\theta},P_n)$ and the method is proved to be robust under Huber Contamination (although this requires a careful selection of the Kernel, which is a difficult functional parameter to tune. The method is further generalized by \cite{pacchiardi2024generalized}, using a general scoring rule approach, which reduces to MMD-Bayes when using an a Kernel scoring rule. Recently \cite{bharti2026amortised} consider $\mathcal{L}(P_n,P_{\theta})$ an unbiased estimate of the weighted score matching divergence between $P_{\theta}$ and $P$, also showing robustness under Huber Contamination. An alterative non-parametric approach to Generalized Bayes is taken by placing a Dirichlet process prior directly on the data generating distribution, $P \sim DP(\mathbf{F},\alpha)$, and then exploiting conjugacy. Specifically, one can generate an induced posterior distribution over a parameter by repeatedly drawing samples from the non-parametric posterior $Q \sim DP(\frac{n}{\alpha+n} P_n + \frac{\alpha}{\alpha+n} \mathbf{F},\alpha+n)$ and finding a minimum distance estimator $\min_{\theta} D(P_{\theta},Q)$. \cite{dellaporta2022robust} take this approach with MMD, showing it is robust to Huber contamination. When $\alpha = 0$, their approach reduces to the Bayesian Bootstrap on top of minimum squared MMD. Achieving uncertainty quantification does not require the Bayesian perspective, and \cite{bernton2019parameter} takes this approach, implementing Efron's percentile bootstrap \citep{efron1987better} on top of vanilla minimum Wasserstein distance. Optimal transport richly characterizes the distance between probability distributions, perfectly respecting translations of probability distributions (a property MMD notoriously does not have), but it is not robust even to Huber contamination.

The quantity $\ell_{\lambda}(P_1,P_2)$ we use for our min distance approach is exactly the Robust Semi-constrained Optimal Transport introduced by \cite{le2021robust}. While we view it as arising from an Empirical Likelihood motivation, they view it as arising through a series of modifications to the Unbalanced Optimal Transport \citep{chizat2018scaling,chizat2017unbalanced} cost. \cite{le2021robust} focuses on its application to Barycenter Computation for discrete distributions. They provide an algorithm for computing $\ell_{\lambda}(P_1,P_2)$ up to $\epsilon$ precision in $O(\frac{n^2}{\epsilon})$ time when $P_1,P_2$ are discrete; note that we address semi-discrete computation, which is naturally relevant in parametric statistical inference, where it is necessary to compare $P_n$ (discrete) with $P_{\theta}$ (which is often absolutely continuous with respect to Lebesgue measure). \cite{wang2024outlier} give Robust Semi-constrained Optimal Transport a different name, calling it Unbalanced Wasserstein Distance, and use it in the context of Distributionally Robust Optimization (DRO). 


Whereas Robust Statistics is concerned with returning an estimate as close to $\theta^{*}$ as possible, DRO addresses the scenario that the decision for the parameter will be deployed in an out-of-sample environment (test distribution), which may slightly differ even from the clean training distribution $P_{\theta^{*}}$. This is generally addressed by constructing a ball of radius $\delta$ around the empirical measure $P_n$, $B_{\delta}(P_n)$, and choosing an estimator (decision)  $\theta^{**}_{\delta} = \inf_{\theta \in \Theta} \sup_{Q \in B_{\delta}(P_n)} \mathbb{E}_{\xi \sim Q} L(\theta,\xi)$ for some loss function $L$. $\delta$ introduces conservatism into the decision making process. $B$ equal to a Wasserstein ball corresponds to Wasserstein DRO \citep{gao2024wasserstein}. In, for example, the strong $\epsilon$-contamination model, as $n \to \infty$, a Wasserstein-2 ball around $P_n$ will not necessarily cover $P_{\theta^{*}}$, because Wasserstein-2 distance has explosive sensitivity to repositioning an $\epsilon$ fraction of mass an extremely far distance. This is one motivation for replacing the Wasserstein ball with a Robust Semi-Constrained Optimal Transport ball as in \cite{wang2024outlier}. But note that our work takes the Robust Statistics perspective,  where parameter recovery with uncertainty quantification, not systematic conservatism to handle out-of-sample distribution shift, is the goal of inference. A complete discussion of the difference between Robust Statistics and DRO is provided by \cite{blanchet2025distributionally}.

We also note that there are other ways of defining Robust Optimal Transport. \cite{mukherjee2021outlier} and \cite{nietert2023robust} consider a Total Variation, rather than Kullback-Liebler derived Robust Optimal Transport and \cite{nietert2023outlier} applies it in the context of DRO.

\section{Proofs of results from main body of text}
\label{sec:proofsInText}

\subsection{Helper lemmas}

For convenience in these appendices, we introduce the further notation $O: \mathbb{R}^{n} \to \mathbb{R}$,
\begin{equation}
\label{defOfO}
O(g) := \mathbb{E}_{X \sim P} h_1(X,g)
\end{equation}
where recall $h_1$ is defined in Equation \ref{h1Def}. Before completing the proofs of the results in the main body of the paper, we state and prove several useful helper lemmas.

We first prove a useful fact about $h_1$. Recall that for a concave function $f: \mathbb{R}^{n} \to \mathbb{R}$, its subdifferential at $g$, denoted $\partial f(g)$, is defined as the set of slope directions tangent to $f$ at $g$ such that the tangent line lies above the function $f$ everywhere. Formally,
\begin{equation}
\label{subdifferential}
\partial f(g) := \{ z \in \mathbb{R}^{n} : \forall \eta \in \mathbb{R}^{n}, f(\eta) \leq f(g) + z^T(\eta - g) \}.
\end{equation}

\begin{lemma}
    \label{helper:subDifferential}
    With $Y_1,\dots,Y_n \in \mathbb{R}^{m}$, for every $x \in \mathbb{R}^{m}$, $h_1(x,\cdot)$ is concave in its second argument. Its sub-differential satisfies for every $x \in \mathbb{R}^{m},g \in \mathbb{R}^{n}$
    \[
      \left\{  \left( \left( \frac{\exp(-\lambda g_1)}{\sum_{t=1}^{n} \exp(-\lambda g_t) }, \dots, \frac{\exp(-\lambda g_n)}{\sum_{t=1}^{n} \exp(-\lambda g_t)}  \right)-\mathbf{e}_{j*(g)} \right) | j \in [n], \argmin_{t \in [n]} \| x - Y_t \|_2^2 - g_t = j \right\}
     \subseteq 
     \partial h_1(x,g)
    \]
    where $\mathbf{e}_{j}$ is the vector with a $1$ at position $j$ and $0$ everywhere else.
\end{lemma}
\begin{proof}
    Fix $x \in \mathbb{R}^{m}$. Using the variational identity of Lemma \ref{variationalIdentity}, observe that $\log(\sum_{j=1}^{n} \frac{1}{n} \exp(- \lambda g_j) )$ is the suprema of convex functions of $g$. Thus $\log(\sum_{j=1}^{n} \frac{1}{n} \exp(- \lambda g_j) )$ is convex in $g$. Thus $-\frac{1}{\lambda} \log(\sum_{j=1}^{n} \frac{1}{n} \exp(- \lambda g_j) )$ is concave in $g$. Also each function of $g$ in the $\min$ in $h_1$ is concave in $g$, and the min of concave functions is concave. Thus the two summands of $h_1$ are concave in $g$, and the sum of concave functions is concave, hence $h_1(x,\cdot)$ is concave in its second argument.

    Regarding the subdifferential, Let $z_1 : \mathbb{R}^{n} \to \mathbb{R}$ be defined as
    \[
    z_1(g) = -\frac{1}{\lambda} \log(\sum_{t=1}^{n} \frac{1}{n} \exp(-\lambda g_t) ).
    \]
    $z_1$ is twice continuously differentiable and
    \[
    \nabla z_1(g) =  \pmb{p}_{g}
    \]
    and
    \[
    \nabla \pmb{p}_{g} = \lambda \left( \pmb{p}_{g} \pmb{p}_{g}^{T} - diag(\pmb{p}_{g}) \right)
    \]
    where $\pmb{p}_{g} = \left( \frac{\exp(-\lambda g_1)}{\sum_{t=1}^{n} \exp(-\lambda g_t) }, \dots, \frac{\exp(-\lambda g_n)}{\sum_{t=1}^{n} \exp(-\lambda g_t)}  \right)^{T}$. In particular, note that for any vector $\pmb{a} \in \mathbb{R}^{n}$, 
    \[
    a^T \nabla \pmb{p}_{g} a = -Var(\sum_{t=1}^{n} \delta_{a_t} \pmb{p}_{g,t}) \leq 0.
    \]
    For any $g \in \mathbb{R}^{n}$, by a two-term Taylor series expansion of $z$ about $g$, we thus have that for any $\eta \in \mathbb{R}^{n}$,
    \begin{equation}
        \label{z1TangentUpper}
    z_1(\eta) \leq z_1(g) + ( \pmb{p}_{g})^T(\eta - g).
    \end{equation}
    Now let $z_2 : \mathbb{R}^{n} \to \mathbb{R}$ be defined as
    \[
    z_2(g) =  \left(\min_{t \in [n]} \| x - Y_t \|_2^2 - g_t\right).
    \]
    For $g \in \mathbb{R}^{n}$, let $j^{*}(g)$ be any index achieving the $\min$ in $z_2$. Then for any $\eta \in \mathbb{R}^{n}$,
    \begin{equation}
        \label{subgradForDiscontinuousPart}
        \begin{split}
    z_2(\eta) =  \| x - Y_{j*(\eta)} \|_2^2 - \eta_{j*(\eta)} \leq & \\
    \| x - Y_{j*(g)} \|_2^2 - \eta_{j*(g)} = & \\
     \| x - Y_{j*(g)} \|_2^2 - g_{j*(g)} + g_{j*(g)} - \eta_{j*(g)} = & \\
     z_2(g) + (-\mathbf{e}_{j*(g)}^T) ( \eta - g).
    \end{split}
    \end{equation}
    For $g \in \mathbb{R}^{n}$, by summing equations \ref{z1TangentUpper} and \ref{subgradForDiscontinuousPart}, we have that 
    \begin{equation}
        z_1(\eta)+z_2(\eta) \leq z_1(g)+z_2(g) + (\pmb{p}_{g}- \mathbf{e}_{j*(g)})^T(\eta - g).
    \end{equation}
    As $z_1(\eta)+z_2(\eta) = h_1(x,g)$ for $g \in \mathbb{R}^{n}$, we conclude by the definition the subdifferential that $ (\pmb{p}_{g} - \mathbf{e}_{j*(g)}) \in \partial h_1(x,g)$.
\end{proof}

One more useful fact about $h_1$ before moving onto studying $O$.
\begin{lemma}
    \label{lipchitzH}
    Let $A \subset \mathbb{R}^{m}$ such that $\sup_{x,y \in A} \|x - y \|_2 \leq D$ and $\mathbf{0}_{m} \in A$ and $Y_1,\dots,Y_n \in A$ and $\lambda >0$. Then for every $g \in \mathbb{R}^{n}$, $h_1(\cdot,g)$ is $3 D$-Lipchitz continuous on $A$.
\end{lemma}
\begin{proof}
Fix $g \in \mathbb{R}^{n}$. Let $f_j(x) = \| x - Y_j \|_2^2 - g_j$ for $j \in [n]$. For $x_1,x_2 \in A$
\begin{equation}
    \begin{split}
        \label{lipchitzFJ}
        |f_j(x_1) - f_j(x_2) | = & \\
        |\| x_1 - Y_j \|_2^2 - \| x_2 - Y_j \|_2^2 | = & \\
        | (x_1-x_2)^{T}(x_1+x_2) - 2(x_1-x_2)^T Y_j | = & \\
        |(x_1-x_2)^T((x_1 - Y_j) +(x_2 - Y_j) ) | \leq & \\
        \| x_1 - x_2 \|_2 \left( \| x_1 - Y_j \|_2 + \| x_2 - Y_j \|_2 \right) \leq & \\
        2D \| x_1 - x_2 \|_2.
    \end{split}
\end{equation}
Additionally, with $F : A \to \mathbb{R}$ defined as $F(x) = \min_{j \in [n]} f_j(x)$, for any $x_1,x_2 \in A$, and $j_1 = \argmin_{j \in [n]} f_j(x_1)$ and $j_2 = \argmin_{j \in [n]} f(x_2)$, we have that
\begin{equation}
    f_{j_1}(x_1) - f_{j_1}(x_2) \leq f_{j_1}(x_1) - f_{j_2}(x_2) \leq f_{j_2}(x_1) - f_{j_2}(x_2).
\end{equation}
And $f_{j_1}(x_1) - f_{j_2}(x_2) = F(x_1)-F(x_2)$. Using this and the above equation, we have that
\[
|F(x_1) - F(x_2) | \leq \max_{j \in [n]} |f_{j}(x_1) - f_{j}(x_2)|.
\]
Applying Equation \ref{lipchitzFJ}, we conclude for any $x_1,x_2 \in A$,
\[
|F(x_1) - F(x_2)| \leq 3 D \| x_1 - x_2 \|_2.
\]
Noting that $h_1(x,g) =  F(x)+const$ where $const$ does not depend on $x$ completes the proof.
\end{proof}

Now we prove a couple of useful facts about $O$ in the next two lemmas.
\begin{lemma}
    \label{coercivityOfO}
    For $P \in \mathcal{P}_2(\mathbb{R}^{m})$ and $Y_1,Y_2,\dots,Y_n \in \mathbb{R}^{m}$ with $B_n := \max_{j \in [n]} \mathbb{E}_{X \sim P} \| X - Y_j \|_2^2$,
    \[
    O(g) \leq (B_n - \| g \|_{\infty})+\frac{\log(n)}{\lambda}.
    \]
    for every $g \in \mathbb{R}^{n}$ such that $g_n = 0$.
\end{lemma}
\begin{proof}
Consider $g \in \mathbb{R}^{n}$ such that $g_n = 0$. There are two possibilities. Either $\| g \|_{\infty} = \max_{t} g_t$ or $\| g \|_{\infty} = -\min_{t} g_t$. If $\| g \|_{\infty}  = \max_{t} g_t$, then letting $i$ be an index achieving the maximum in $g$, we have
\begin{equation}
    \label{boundedObjectiveA}
    O(g) \leq  \left(\mathbb{E}_{X \sim P} \| X - Y_i \|_2^2 - g_i \right)-\frac{1}{\lambda}\log(\frac{1}{n} \sum_{t=1}^{n} \exp(-\lambda g_t) ) \leq (B_n -\| g \|_{\infty}) +\frac{\log(n)}{\lambda}.
\end{equation}
If $\| g \|_{\infty} = -\min_{t} g_t$, then we have that
\begin{equation}
    \label{boundedObjectiveB}
    O(g) \leq  \left(\mathbb{E}_{X \sim P} \|X - Y_n \|_2^2\right) - \frac{1}{\lambda} \log(\frac{1}{n} \sum_{t=1}^{n} \exp(-\lambda g_t) ) \leq  B_n - \log(\exp(-\lambda \min_{t \in [n]} g_t) ) \leq  (B_n - \| g \|_{\infty})+\frac{\log(n)}{\lambda}.
\end{equation}
Combining the two cases yields the final result.

\end{proof}

\begin{lemma}
    \label{lem:infoAboutOurObjective}
    With $Y_1,\dots,Y_n \in \mathbb{R}^{m}$, and $P \in \mathcal{P}_2(\mathbb{R}^{m})$ such that $P$ is absolutely continuous with respect to Lebesgue measure, the function $O(g) := \mathbb{E}_{X \sim P} h_1(X,g)$. satisfies the following:
    \begin{enumerate}
        \item $\nabla O(g) = \left( \left( \frac{\exp(-\lambda g_1)}{\sum_{t=1}^{n} \exp(-\lambda g_t)}, \dots, \frac{\exp(-\lambda g_n)}{\sum_{t=1}^{n} \exp(-\lambda g_t)} \right)^T - \left(P(L_1(g)),P(L_2(g)),\dots,P(L_n(g)) \right)^T \right)$ for $g \in \mathbb{R}^{n}$ where $L_j(g) = \{ x : \forall i \in [n], \| x - Y_j \|_2^2 - g_j \leq \| x - Y_i \|_2^2 - g_i\}$ for $j \in [n]$.
        \item $O$ satisfies for any $g,\eta \in \mathbb{R}^{n}$,
        \[
        O(\eta) \leq O(g) + \nabla O(g)^{T}(\eta - g).
        \]
        \item Any $g^*$ achieves the maximium $\max_{g \in \mathbb{R}^{n}} O(g)$ if and only if $\nabla O(g^*) = 0 $. Additionally, there exists at least one such $g^{*}$.
        \item There exists a $g^{*} \in \mathbb{R}^{n}$ achieving the maximum $\max_{g \in \mathbb{R}^{n}} O(g)$ such that $g^{*}_n = 0$. For such a $g^{*}$, if $Y_1,\dots,Y_n \in A \subset \mathbb{R}^{m}$ and $P \in \mathcal{P}_{2}(A)$ and $\sup_{x,y \in A} \| x - y\|_2 \leq D$ where $D > 0$, then
        \[
        \|g^{*} \|_{\infty} \leq  D^2.
        \]
    \end{enumerate}
    
\end{lemma}
\begin{proof}
    For the function $Z : \mathbb{R}^{n} \to \mathbb{R}$ satisfying
    \[
    Z(g) := \mathbb{E}_{X \sim P} \min_{j \in [n]}\|X - Y_j \|_2^2 - g_j 
    \]
    Let $\delta \in \mathbb{R}$ and define for $g \in \mathbb{R}^{n}$, $g_{j'}^{\delta} := (g_1,g_2,\dots,g_{j'-1},g_{j'}+\delta,g_{j'+1},\dots,g_{n})$ for $j' \in [n]$. Let $S := \{x \in \mathbb{R}^{m}: \argmin_{j \in [n]} \|x -Y_j\|_2
    ^2 - g_j \text{ is unique} \}$. We have that
    \begin{equation}
        \lim_{\delta \to 0} \frac{\min_{j} \| x - Y_j \|_2^2 -(g_j+\delta \mathbb{I}(j=j')) - \min_{j} \|x - Y_j\|_2^2 - g_j}{\delta} = \begin{cases} 
            -1 & x \in S \cap L_{j'}(g) \\
            0 & x \in S \cap \left(\mathbb{R}^{m} \setminus L_{j'}(g) \right).
        \end{cases}
    \end{equation} 
    The term in the limit is also uniformily over $\delta \in \mathbb{R}$ dominated by $1$, since the difference between the mins is never more in absolute value than $\delta$. Thus by DCT and using that $S$ has Lebesgue measure $1$, we have that
    \begin{equation}
        \begin{split}
        \lim_{\delta \to 0} \frac{ z(g_{j'}^{\delta}) - z(g)}{\delta} = \lim_{\delta \to 0} \mathbb{E}_{X \sim P} \frac{\min_{j} \| X - Y_j \|_2^2 -(g_j+\delta \mathbb{I}(j=j')) - \min_{j} \|X - Y_j\|_2^2 - g_j}{\delta} = & \\
        \mathbb{E}_{X \sim P} \mathbb{I}(X \in S) \lim_{\delta \to 0} \frac{\min_{j} \| X - Y_j \|_2^2 -(g_j+\delta \mathbb{I}(j=j')) - \min_{j} \|X - Y_j\|_2^2 - g_j}{\delta} = & \\
        -P(L_{j'}(g)).
        \end{split}
    \end{equation}
    In particular, for $j' \in [n],$
    \[
    \nabla Z(g)_{j'} = -P(L_{j'}(g)).
    \]
    The gradient of $\log(\frac{1}{n} \sum_{j=1}^{n} \exp(-\lambda g_j))$ is obtainable through elementary calculus. Adding $\nabla Z$ to the gradient of the log-sum-exp term yields (1). \\

    For (2), by Lemma \ref{helper:subDifferential}, for every $x \in \mathbb{R}^{m}$ and $g \in \mathbb{R}^{n}$, we have that for every $\eta \in \mathbb{R}^{n}$,
    \begin{equation}
        \label{subGradConclusion}
        h_1(x,\eta) \leq h_1(x,g) + z(x,g)^T (\eta - g).
    \end{equation}
    where
    \[
    z(x,g) :=  \left( \left( \frac{\exp(-\lambda g_1)}{\sum_{t=1}^{n} \exp(-\lambda g_t) }, \dots, \frac{\exp(-\lambda g_n)}{\sum_{t=1}^{n} \exp(-\lambda g_t)}  \right)-\mathbf{e}_{j^*(x,g)} \right)^T
    \]
    where $j^*(x,g)$ refers to lowest index achieving the min in $\min_{j \in [n]} \|x - Y_j\|_2^2 - g_j$. Note that because $P$ is absolutely continuous with respect to Lebesgue measure, for $j \in [n]$, the $j^{th}$ entry of the random vector $\mathbf{e}_{j^*(X,g)}$ is $1$ with probability equal to the probability that $X \in L_j(g)$, and $0$ otherwise when $X \sim P$. In particular, for $g \in \mathbb{R}^{n}$,
    \[
    \mathbb{E}_{X \sim p} z(X,g) = \left( \frac{\exp(-\lambda g_1)}{\sum_{t=1}^{n} \exp(-\lambda g_t) }, \dots, \frac{\exp(-\lambda g_n)}{\sum_{t=1}^{n} \exp(-\lambda g_t)}  \right)^{T}-(P(L_1(g)),\dots,P(L_n(g))^T.
    \]
    Using this and (1) and Equation \ref{subGradConclusion}, we have that for every $g \in \mathbb{R}^{n}$ and $\eta \in \mathbb{R}^{n}$,
    \begin{equation}
        O(\eta) \leq O(g) + \nabla O(g)^{T}(\eta - g),
    \end{equation}
    concluding (2). For (3), if we suppose that $g \in \mathbb{R}^{n}$ satisfies that $\nabla_{j} O(g) \neq 0$ for some $j$, then by (1), for some $r(\delta) \to 0$ as $\delta \to 0$,
    \[
    \frac{O(g + \mathbf{e}_{j} \delta ) - O(g)}{\delta} - \nabla_{j} O(g)= r(\delta).
    \]
    In particular, choosing $\delta$ to have the same sign as $\nabla_{j} O(g)$ and then taking $\delta$ sufficiently close to zero so that $\delta r(\delta) \geq -\frac{\delta}{2} \nabla_{j} O(g)$, we have that
    \[
    O(g+ \mathbf{e}_{j}\delta) = O(g)+\delta (\nabla_{j} O(g) + r(\delta)) \geq O(g)+ \frac{\delta}{2} \nabla_{j} O(g) \geq O(g).
    \]
    This shows that if $g$ maximizes $O$ on $\mathbb{R}^{n}$, then $\nabla O(g) = 0$. By (2), we also have that if $\nabla O(g) = 0$, then $g$ maximizes $O$ on $\mathbb{R}^{n}$. Thus $g$ maximizes $O$ on $\mathbb{R}^{n}$ if and only if $\nabla O(g) = 0$. All that is left to show for (3) is that there exists at least one maximizer.
    
    Let $\mathcal{G}_{0} := \{ g \in \mathbb{R}^{n} : g_n  = 0\}$. By Lemma \ref{coercivityOfO}, there exists an $M >0$ such that whenever $\| g \|_{\infty} > M$ and $g \in \mathcal{G}_{0}$, $O(g) \leq O(\mathbf{0}_n)$. In particular, if there is a maximizer of $O$ on $\mathbb{R}^{n} \cap \mathcal{G}_{0}$ it is within the set $\mathcal{G}_{0} \cap \{g : \| g \|_{\infty} \leq M\}$. Continuous functions on compact sets achieve their Suprema (Theorem 4.16 of \cite{rudin1976principles}), hence $O$ achieves a local maxima $g^{*}$ on $\mathcal{G}_{0} \cap \{g : \| g \|_{\infty} \leq M\}$ which consequently is maximal on $\mathbb{R}^{n} \cap \mathcal{G}_{0}$. Specifically,
    \[
    O(g^{*}) = \max_{g \in \mathbb{R}^{n} \cap \mathcal{G}_{0}} O(g).
    \]
    Finally note that for any vector $g \in \mathbb{R}^{n}$ and constant $c \in \mathbb{R}$, $O(g) = O(g+c\mathbf{1}_{n})$. This invariance implies for any $g \in \mathbb{R}^{n}$, there is a $g' \in \mathbb{R}^{n} \cap \mathcal{G}_{0}$ such that $O(g) = O(g')$. Thus $\max_{g \in \mathbb{R}^{n}} O(g) \leq \max_{g \in \mathbb{R}{n} \cap \mathcal{G}_{0}} O(g) $. This together with $\mathbb{R}^{n} \cap \mathcal{G}_{0} \subseteq \mathbb{R}^{n}$ implies
    \[
    \max_{g \in \mathbb{R}^{n} \cap \mathcal{G}_{0}} O(g) = \max_{g \in \mathbb{R}^{n}} O(g), 
    \]
    hence a maximizer of $O$ exists on $\mathbb{R}^{n}$. This completes (3).

    For (4), via (3) and the translation invariance, there must exist a $g^{*} \in \mathcal{G}_{0}$ achieving the maximum. Additionally, by (3) this $g^{*}$ must satisfy the zero gradient condition. In particular, for $j \in [n]$,
    \begin{equation}
        \label{firstOrderCondition}
        P(L_j(g^{*})) = \frac{\exp(-\lambda g_j^{*})}{\sum_{t=1}^{n} \exp(-\lambda g_t^{*})}.
    \end{equation}
    Now let $i$ be the index of $g^{*}$ achieving $\| g^{*} \|_{\infty}$. 
    Suppose for sake of contradiction, $\| g^{*} \|_{\infty} > D^2$. Then if $\| g^{*} \|_{\infty} > 0$, for every $x \in A$,
    \[
    \| x - Y_i \|_2^2 - g_i \leq D^2 - D^2 = 0 \leq \| x - Y_n \|_2^2,
    \]
    implying $P(L_n(g^{*})) = 0$, which cannot be true via Equation \ref{firstOrderCondition}. Alternatively if $\| g \|_{\infty} \leq 0$, then for every $x \in A$
    \[
    \| x - Y_n \|_2^2 \leq D^2 \leq - g_i \leq \| x - Y_i \|_2^2 - g_i,
    \]
    implying $P(L_i(g^{*})) = 0$, which cannot be true via Equation \ref{firstOrderCondition}. Thus we conclude that $\| g^{*} \|_{\infty} \leq D^2$.
    
\end{proof}

Although the below fact regarding Kantorovich's functional (see \cite{kitagawa2019convergence} for a deeper dive) is extremely well known, for the sake of completeness we state it.

\begin{lemma}
    \label{kantorovichsFunctional}
    With $Y_1,\dots,Y_n \in \mathbb{R}^{m}$ and $P \in \mathcal{P}_{2}(\mathbb{R}^{m})$ such that $P$ is absolutely continuous with respect to Lebesgue measure, and $w \in \Sigma_n$, then for the function
    \[
    O_2(g) = \mathbb{E}_{X \sim P} \left(\min_{j \in [n]} \|X - Y_j \|_2^2 - g_j\right) + \sum_{t=1}^{n} w_j g_j \text{ (a.k.a Kantorovich's Functional) } 
    \]
    \begin{enumerate}
        \item $\nabla O_2(g) = \left(w_1-P(L_1(g)),\dots,w_n - P(L_n(g)) \right)^T$
        \item Any $g^{*}$ achieves the maximum $\max_{g \in \mathbb{R}^{n}} O_2(g)$ if and only if $\nabla O_2(g) = 0$. Additionally, there exists at least one such $g^{*}$.
        \item $W_2^2(P,\sum_{t=1}^{n} w_t \delta_{Y_t}) = \max_{g \in \mathbb{R}^{n}} O_2(g)$.
    \end{enumerate}
\end{lemma}
\begin{proof}
    The first two statements are proved using the same strategy as in the proof of Lemma \ref{lem:infoAboutOurObjective}. For part 3, see for example Chapter 5, Section 1 of \cite{peyre2019computational}.
\end{proof}

\subsection{Proofs of lemmas from section \ref{sec:basicProperties}}



\divergence*
\begin{proof}
The non-negativity follows since $\mathrm{KL}$ and $W_2^2$ are both divergences. Suppose $P_1 = P_2$. Set $Q = P_1$. Then both the $\mathrm{KL}$ and $W_2^2$ terms are zero. Now suppose $\ell_{\lambda}(P_1,P_2) = 0$. Then there exists a sequence of measures $\{Q_i\}_{i=1}^{\infty}$ (all absolutely continuous with respect to $P_1$) such that $\lim_{i \to \infty} \frac{1}{\lambda} \mathrm{KL}(Q_i,P_1)+ W_2^2(P_2,Q_i) = 0$. This implies $\lim_{i \to \infty} \mathrm{KL}(Q_i,P_1) = 0$. By Pinsker's inequality, $\lim_{i \to \infty} \mathrm{TV}(Q_i,P_1) = 0$. By the Portmanteau Lemma (see for example Theorem 13.16 of \cite{klenke2008probability}), $Q_i$ converges weakly to $P_1$. Also $\lim_{i \to \infty} W_2^2(P_2,Q_i) = 0$. Thus, by \cite{villani2009optimal} Theorem 6.9, $Q_i$ converges weakly to $P_2$. Since the weak limit of a sequence of measures in $\mathcal{P}_{2}(\mathbb{R}^{m})$ is unique, $P_1 = P_2$.   
\end{proof}
\smallLam*
\begin{proof}
    Let $\epsilon > 0$. Consider any sequence of measures $\{Q_{\lambda,\epsilon}\}_{\lambda >0}$ such that 
    \begin{equation}
        \label{QlamDef}
    \frac{1}{\lambda}\mathrm{KL}(Q_{\lambda,\epsilon},P_1)+W_2^2(P_2,Q_{\lambda,\epsilon}) \leq  \ell_{\lambda}(P_1,P_2)+\epsilon
    \end{equation}
    By setting $Q = P_1$, for each $\lambda > 0$, $ \ell_{\lambda}(P_1,P_2) \leq W_2^2(P_1,P_2)$. In particular
    \[
    \frac{1}{\lambda} \mathrm{KL}(Q_{\lambda,\epsilon},P_1) \leq \frac{1}{\lambda}\mathrm{KL}(Q_{\lambda,\epsilon},P_1)+W_2^2(P_2,Q_{\lambda,\epsilon}) \leq W_2^2(P_1,P_2)+\epsilon.
    \]
    Thus for every $\lambda > 0$,
    \[
    \mathrm{KL}(Q_{\lambda,\epsilon},P_1) \leq \lambda W_2^2(P_1,P_2)+\lambda \epsilon.
    \]
    Taking $\lambda \to 0$ yields, for every $\epsilon >0$
    \begin{equation}
        \label{Qlim}
    \lim_{\lambda \to 0} \mathrm{KL}(Q_{\lambda,\epsilon},P_1) = 0.
    \end{equation}
    Since $Q_{\lambda,\epsilon}$ and $P_1$ are defined on bounded diameter spaces, $W_2^2(Q_{\lambda,\epsilon},P_1) \leq C \mathrm{TV}(Q_{\lambda,\epsilon},P_1) \leq C \sqrt{\mathrm{KL}(Q_{\lambda,\epsilon},P_1)}$ where the last inequality is due to Pinsker (and here $C>0$ represents some positive constant). Using this and Equation \ref{Qlim}, we conclude that
    \begin{equation}
        \label{w2LimA}
    \lim_{\lambda \to 0} W_2(Q_{\lambda,\epsilon},P_1) = 0.
    \end{equation}
    By triangle inequality for $W_2$
    \[
    W_2(P_1,P_2) - W_2(P_1,Q_{\lambda,\epsilon}) \leq W_2(Q_{\lambda,\epsilon},P_2)
    \]
    Using this and Equation \ref{w2LimA}, we have for every $\epsilon > 0$
    \[
    W_2^2(P_1,P_2) \leq \liminf_ { \lambda \to 0}W_2^2(Q_{\lambda,\epsilon},P_2)
    \]
    Thus for $\epsilon >0$
    \[
    W_2^2(P_1,P_2) \leq \liminf_{\lambda \to 0} \frac{1}{\lambda} \mathrm{KL}(Q_{\lambda,\epsilon},P_1)+W_2^2(Q_{\lambda,\epsilon},P_2) 
    \]
    Finally, applying Equation \ref{QlamDef} for every $\lambda,\epsilon >0$ yields, for every $\epsilon >0$,
    \begin{equation}
        W_2^2(P_1,P_2) \leq \liminf_{\lambda \to 0}  \ell_{\lambda}(P_1,P_2) + \epsilon
    \end{equation}
    This holds for every $\epsilon >0$, and so we have that
    \[
    W_2^2(P_1,P_2) \leq \liminf_{\lambda \to 0} \ell_{\lambda}(P_1,P_2) \leq \limsup_{\lambda \to 0}  \ell_{\lambda}(P_1,P_2) \leq W_2^2(P_1,P_2)
    \]
    where the upper bound again follows by setting $Q = P_1$.
\end{proof}

Before proving the next lemma, we need the following foundational lemma.

\begin{lemma}
    \label{voronoiWeightingIsBest}
    Let $y_1,y_2,\dots,y_n \in \mathbb{R}^{m}$, $P_2 \in \mathcal{P}_2(\mathbb{R}^{m})$ absolutely continuous with respect to Lebesgue Measure,  $\pmb{w}^{*} = (w_1^{*},w_2^{*},\dots,w_n^{*})$ where for $i \in [n]$,
    \[
    w_i^{*} = P_2(x \in \mathbb{R}^{m}: \| x - y_i\|_2^2 \leq \| x - y_j \|_2^2, \forall j \in [n]).
    \]
    Then defining the map $L: \Delta^{n-1} \to \mathbb{R}$ as $L(\pmb{w}) = W_2^2(\sum_{i=1}^{n} w_i \delta_{y_i},P_2)$, we have that
    \[
    \argmin_{\pmb{w} \in \Delta^{n-1}} L(\pmb{w}) = \pmb{w}^{*}
    \]
    and for every $\pmb{w} \neq \pmb{w}^{*}$, $L(\pmb{w}) > L(\pmb{w}^{*})$. Also $L$ is continuous on $\Delta^{n-1}$ with respect to Euclidean distance.
\end{lemma}
\begin{proof}
    The first order optimality condition of the semi-dual representation of $W_2^2(\sum_{i=1}^{n} w_i^{*} \delta_{y_i},P_2) = \max_{\pmb{g} \in \mathbb{R}^{n}}  \int_{x \in \mathbb{R}^{m}} \min_{j \in [n]} \left( \| x - y_j \|_2^2 - g_j \right)P_2(dx)+ \sum_{j=1}^{n} w_j^{*} g_j$ is that 
    \[
    w_i^{*} = P_2(L_i(\pmb{g}))
    \]
    where $\pmb{g} \in \mathbb{R}^{n}$ and $L_j$ is the $j^{th}$ Laguerre cell of $\pmb{g}$ (see for example \cite{peyre2019computational} Section 5.2). By definition of $w_i^{*}$ and the Laguerre cells, this implies an optimal $\pmb{g}$ is the zero vector. Plugging this $\pmb{g}$ into the semi-dual yields
    \[
    L(\pmb{w}^{*}) = W_2^2(\sum_{i=1}^{n} w_i^{*} \delta_{y_i},P_2) = \int_{x \in \mathbb{R}^{m}} \| x - N(x) \|_2^2 P_2(dx)
    \]
    where $N$ is the map that takes $x$ to its nearest neighbor in the set $y_1,y_2,\dots,y_n$. By Bernier's theorem (see \cite{chewi2025statistical} Theorem 1.8), for every $\pmb{w} \in \Delta^{n-1}$, there exists a map $T_{\pmb{w}}: \mathbb{R}^{m} \to \{y_1,y_2,\dots,y_n\}$ satisfying $P_2(x \in \mathbb{R}^{m}: T(x) = y_j) = w_j$ for each $j \in [n]$ and
    \[
    L(\pmb{w}) = W_2^2(\sum_{i=1}^{n} w_i \delta_{y_i},P_2) = \int_{x \in \mathbb{R}^{m}} \| x - T_{\pmb{w}}(x) \|_2^2 P_2(dx).
    \]
    Thus for every $\pmb{w} \in \Delta^{n-1}$,
    \[
    L(\pmb{w}) - L(\pmb{w}^{*}) = \int_{x \in \mathbb{R}^{m}} \| x - T_{\pmb{w}}(x) \|_2^2 - \| x - N(x) \|_2^2 P_2(dx)
    \]
    Now suppose there exists a $\pmb{w} \neq \pmb{w}^{*}$ such that the Lebesgue measure of $\mathcal{S}_{\pmb{w}} = \{x \in \mathbb{R}^{m}: N(x) \neq T_{\pmb{w}}(x) \} = 0$. Then since $P_2$ is absolutely continuous with respect to Lebesgue measure, for every $j$, $w_j = P_2(x \in \mathbb{R}^{m}: T(x) = y_j) = P_2(x \in \mathbb{R}^{m}: N(x) = y_j) = w_j^{*}$. This contradiction implies that if $\pmb{w} \neq \pmb{w}^{*}$, then the lebesgue measure of $\mathcal{S}_{\pmb{w}}$ is positive. In particular for $\pmb{w} \neq \pmb{w}^{*}$
    \begin{equation}
        L(\pmb{w}) - L(\pmb{w}^{*}) = \int_{x \in \mathcal{S}_{\pmb{w}}} \| x - T_{\pmb{w}}(x) \|_2^2- \| x - N(x) \|_2^2 P_2(dx) >0
    \end{equation}
    where the $>0$ is since $P_2$ has is absolutely continuous with respect to Lebesgue measure, $\mathcal{S}_{\pmb{w}}$ has positive Lebesgue measure, and $N$ is nearest neighbor assignment, which implies that $\|x - T_{\pmb{w}}(x) \|_2^2 > \|x - N(x) \|_2^2$ for every $x \in \mathcal{S}_{\pmb{w}}$.

    Finally, regarding continuity of $L$ at $\pmb{w}^{*}$, first note for a sequence $\{\pmb{w}_{\ell} \}_{\ell >0} \in \Delta^{n-1}$ and $\pmb{\mu} \in \Delta^{n-1}$
    \begin{equation}
        \label{continuityOfLpenultimate}
    |W_2(\sum_{i=1}^{n} w_{\ell,i} \delta_{y_i}, P_2) - W_2(\sum_{i=1}^{n} \mu_{i} \delta_{y_i},P_2)| \leq W_2(\sum_{i=1}^{n} w_{\ell, i} \delta_{y_i} , \sum_{i=1}^{n} \mu_{i} \delta_{y_i})
    \end{equation}
    Now note that for any continuous and bounded function $f: \mathbb{R}^{m} \to \mathbb{R}$, if $\| \pmb{w}_{\ell} - \pmb{\mu} \|_2 \to 0$ as $\ell \to \infty$, then
    \[
    |\sum_{i=1}^{n} w_{\ell,i} f(y_i) - \sum_{i=1}^{n} \mu_i f(y_i)| \leq C \| \pmb{w}_{\ell} - \pmb{\mu} \|_1 \leq nC \| \pmb{w}_{\ell} - \pmb{\mu} \|_2 \to 0.
    \]
    as $\ell \to \infty$ where $C$ is the upper bound on $f$. Additionally, if $\| \pmb{w}_{\ell}-\pmb{\mu} \|_2 \to 0$ as $\ell \to \infty$, then
    \[
    |\sum_{i=1}^{n} w_{\ell,i} \| y_i \|_2^2 - \sum_{i=1}^{n} \mu_i \| y_i \|_2^2 | \leq n \max_{i} \| y_i \|_2^2 \| \pmb{w}_{\ell} - \pmb{\mu} \|_2 \to 0.
    \]
    as $\ell \to \infty$. This establishes that the sequence of probability measures $\sum_{i=1}^{n} w_{\ell,i} \delta_{y_i}$ converges weakly in $\mathcal{P}_2(\mathbb{R}^{m})$ to $\sum_{i=1}^{n} \mu_i \delta_{y_i}$ as $\ell \to \infty$.
    By \cite{villani2009optimal} Theorem 6.9, we thus have that $W_2(\sum_{i=1}^{n} w_{\ell,i} \delta_{y_i}, \sum_{i=1}^{n} \mu_i \delta_{y_i}) \to 0$ as $\ell \to \infty$ whenever $\| \pmb{w}_{\ell} - \pmb{\mu} \|_2 \to 0$ as $\ell \to \infty$. By Equation \ref{continuityOfLpenultimate} and continuity of the square function, we conclude that whenever $\| \pmb{w}_{\ell} - \pmb{\mu} \|_2 \to 0$ as $\ell \to \infty$ for a sequence $\{\pmb{w}_{\ell} \}_{\ell >0} \in \Delta^{n}$, then
    \[
    |L(\pmb{w}_{\ell,i}) - L(\pmb{\mu})| = |W_2^2(\sum_{i=1}^{n} w_{\ell,i} \delta_{y_i}, P_2) - W_2^2(\sum_{i=1}^{n} \mu_{i} \delta_{y_i},P_2)| \to 0
    \]
    as $\ell \to \infty$. In particular since $\pmb{\mu} \in \Delta^{n-1}$ was arbitrary, $L$ is continuous on $\Delta^{n-1}$ with respect to Euclidean distance.
\end{proof}

\largeLamDiscrete*
\begin{proof}
    Let $Q^{*}_{\infty}$ be the probability measure supported on $y_1,y_2,\dots,y_n$ with weights $w_1^{*},\dots,w_{n}^{*}$ satisfying for $j \in [n]$,
    \[
    w_j^{*} := P_2 \left( \{x \in \mathbb{R}^{m} : \| x - y_j \|_2^2 \leq \| x - y_t \|_2^2, \forall t \} \right)
    \]
    Let $\epsilon > 0$. For $\lambda >0$, by the definition of $\ell_{\lambda}$, there exists a $\pmb{w}_{\lambda,\epsilon} = (w_{1,\lambda,\epsilon},\dots,w_{n ,\lambda,\epsilon}) \in \Delta^{n-1}$ such that the probability measure $Q_{\lambda,\epsilon} = \sum_{i=1}^{n} w_{i ,\lambda,\epsilon} \delta_{y_i}$ satisfies
    \begin{equation}
    \label{orderingForLargeLamDiscrete}
        \begin{split}
     W_2^2(P_2,Q_{\lambda,\epsilon}) \leq & \\
    \frac{1}{\lambda} \mathrm{KL}(Q_{\lambda,\epsilon},P_n) +  W_2^2(P_2,Q_{\lambda,\epsilon}) \leq & \\
    \ell_{\lambda}(P_n,P_2) + \epsilon \leq & \\
    \frac{1}{\lambda} \mathrm{KL}(Q^{*}_{\infty},P_n) +  W_2^2(P_2,Q^{*}_{\infty})+\epsilon
        \end{split}
    \end{equation}
    By Lemma \ref{voronoiWeightingIsBest}, we additionally have that
    \begin{equation}
        \label{tempString}
    W_2^2(P_2,Q_{\infty}^{*}) \leq \liminf_{\lambda \to \infty} W_2^2(P_2,Q_{\lambda,\epsilon}) \leq \limsup_{\lambda \to \infty} W_2^2(P_2,Q_{\lambda,\epsilon}) \leq \limsup_{\lambda \to \infty} W_2^2(P_2,Q_{\lambda,\epsilon}) + \frac{1}{\lambda} \mathrm{KL}(Q_{\lambda,\epsilon},P_n).
    \end{equation}
    Now applying Equation \ref{orderingForLargeLamDiscrete}, we further have that
    \begin{equation}
        \label{tempString2}
        \begin{split}
        \limsup_{\lambda \to \infty} W_2^2(P_2,Q_{\lambda,\epsilon}) + \frac{1}{\lambda} \mathrm{KL}(Q_{\lambda,\epsilon},P_n) \leq \liminf_{\lambda \to \infty} \ell_{\lambda}(P_n,P_2)+\epsilon \leq \limsup_{\lambda \to \infty} \ell_{\lambda}(P_n,P_2) + \epsilon \leq & \\
        \liminf_{\lambda \to \infty} \frac{1}{\lambda} \mathrm{KL}(Q_{\infty}^{*},P_n)+W_2^2(P_2,Q_{\infty}^{*}) = & \\
        W_2^2(P_2,Q_{\infty}^{*})+\epsilon
        \end{split}
    \end{equation}
    Combining equations \ref{tempString} and \ref{tempString2} yields for 
    $\epsilon >0$
    \begin{equation}
        W_2^2(P_2,Q^{*}_{\infty}) \leq \liminf_{\lambda \to \infty} \ell_{\lambda}(P_n,P_2) + \epsilon \leq \limsup_{\lambda \to \infty} \ell_{\lambda}(P_n,P_2) + \epsilon \leq W_2^2(P_2,Q_{\infty}^{*})+\epsilon.
    \end{equation}
    Since this holds for every $\epsilon >0$ we conclude
    \[
    \lim_{\lambda \to \infty} \ell_{\lambda}(P_n,P_2) = W_2^2(P_2,Q_{\infty}^{*}).
    \]
    Finally note that the optimal $g$ in the kantorovich's optimization for the problem $W_2^2(P_2,Q^{*}_{\infty})$ is $g = \mathbf{0}_{n}$.

\end{proof}

\RobustnessUnderContaminationModel*
\begin{proof}
First note that both $Q \ll P$ and $F \ll P$ due to the the definition of $P$. Let $q = \frac{dQ}{dP}$ and $f = \frac{dF}{dP}$ be the Radon-Nikodym derivatives of $Q$ and $F$ with respect to $P$. For every $P$ measurable set $A \subset \mathbb{R}^{m}$, we have
\begin{equation}
    \mathbb{E}_{X \sim P} \mathbb{I}(X \in A) = (1-\epsilon) \int_{A}q(x) P(dx) + \epsilon \int_{A} f(x) P(dx).
\end{equation}
Thus for every $P$ measurable $A \subseteq \mathbb{R}^{m}$,
\[
\int_{A} 1 - \left((1-\epsilon) q(x) +  \epsilon f(x) \right) P(dx) = 0.
\]
Hence we conclude $P$ almost surely
\[
1 = (1-\epsilon) q(x) + \epsilon f(x).
\]
In particular, $P$ almost surely,
\begin{equation}
    q(x) = \frac{1 - \epsilon f(x)}{1-\epsilon} \leq \frac{1}{1-\epsilon}
\end{equation}
Since $q = \frac{dQ}{dP}$, by defintion of $\mathrm{KL}$, we have that
\begin{equation}
    \mathrm{KL}(Q,P) \leq -\log(1-\epsilon) = \epsilon + \frac{1}{2\zeta^2} \epsilon^2.
\end{equation}
where for the last inequality we used Taylor's theorem and $\zeta \in [1-\epsilon,1]$. Since $\epsilon \leq 1 - \frac{1}{\sqrt{2}}$, we further have
\[
\mathrm{KL}(Q,P) \leq \epsilon + \epsilon^{2}.
\]
Thus 
\begin{equation}
    \label{robustnessB}
\ell_{\lambda}(P,P_{\theta^{*}}) \leq \frac{1}{\lambda} \mathrm{KL}(Q,P) + W_2^2(Q,P_{\theta^{*}}) \leq \frac{\epsilon+\epsilon^2}{\lambda}+\rho^2.
\end{equation}
which also proves $\theta^{*} \in \Omega$. To show $\theta_{\lambda,\eta} \in \Omega$, furthermore note that
\begin{equation}
\label{robustnessPartA}
\inf_{\theta \in \Theta} \ell_{\lambda}(P,P_{\theta}) \leq \ell_{\lambda}(P,P_{\theta^{*}}) \leq \frac{\epsilon+\epsilon^2}{\lambda}+\rho^{2} + \inf_{\theta \in \Theta} \ell_{\lambda}(P,P_{\theta}).
\end{equation}
In this last line we used Lemma \ref{prop:positiveDefinite}, which guarantees $\inf_{\theta \in \Theta} \ell_{\lambda}(P,P_{\theta}) \geq 0$. What is left to show is that $\theta_{\lambda,\eta} \in \Omega$. By definition \ref{def:MinDivergenceEstimator}, we have for $\eta > 0$ and by Equations \ref{robustnessB} and \ref{robustnessPartA},
\begin{equation}
    \ell_{\lambda}(P,P_{\theta_{\lambda,\eta}}) \leq \inf_{\theta \in \Theta} \ell_{\lambda}(P,P_{\theta}) + \eta \leq \frac{\epsilon+\epsilon^2}{\lambda}+\rho^2 + \eta.
\end{equation}
Multiplying by $\lambda$ and using the definition of $\ell_{\lambda}(P,P_{\theta_{\lambda,\eta}})$, it follows $\theta_{\lambda,\eta} \in \Omega$.
\end{proof}

\subsection{Proofs of lemmas from section \ref{sec:algorithms}}

We'll use a variational identity for the Kullback-Liebler divergence. 
\begin{lemma}
    \label{variationalIdentity}
    Let $\pmb{y} = (Y_1,\dots,Y_n) \in \mathbb{R}^{n}$, $\pmb{f} = (f_1,\dots,f_n) \in \mathbb{R}^{n}$. Then 
    \[
    \log \left( \sum_{j=1}^{n} \frac{1}{n} \exp(f_j) \right) = \sup_{\pmb{w} \in \Sigma_n} \left( \sum_{j=1}^{n} f_j w_j - \mathrm{KL}(\sum_{i=1}^{n} \delta_{Y_j} w_j, \sum_{j=1}^{n} \frac{1}{n} \delta_{Y_j}) \right)
    \]
    Additionally, the suprema is achieved with
    \[
    w_j^{*} := \frac{\exp(f_j)}{\sum_{t=1}^{n} \exp(f_j)}
    \]
    for $j \in [n]$.
\end{lemma}
\begin{proof}
Let $H = \sum_{t=1}^{n} \exp(f_j)$. For any $\pmb{w} \in \Sigma_n$, using the definition of $w_j^*$, we have
\begin{equation}
    \begin{split}
    \mathrm{KL}(\sum_{j=1}^{n} \delta_{Y_j} w_j, \sum_{j=1}^{n} \delta_{Y_j} w_j^{*}) - \mathrm{KL}(\sum_{j=1}^{n} \delta_{Y_j} w_j, \sum_{i=1}^{n} \delta_{Y_j} \frac{1}{n}) = & \\
    \sum_{j=1}^{n} w_j \left(\log(\frac{w_j}{w_j^*}) - \log(n w_j)\right) = & \\
    \sum_{j=1}^{n} w_j \left( \log(w_j) - f_j + \log(H) -\log(n) - \log(w_j)\right) = & \\
    - \sum_{j=1}^{n} w_j f_j + \log(\frac{H}{n})
    \end{split}
\end{equation}
Rearranging terms yields that for any $\pmb{w} \in \Sigma_n$,
\begin{equation}
    \label{upperBoundonSup}
\sum_{j=1}^{n} w_j f_j - \mathrm{KL}(\sum_{j=1}^{n} \delta_{Y_j} w_j, \sum_{i=1}^{n} \delta_{Y_j} \frac{1}{n}) = \log(\frac{H}{n}) - \mathrm{KL}(\sum_{j=1}^{n} \delta_{Y_j} w_j, \sum_{j=1}^{n} \delta_{Y_j} w_j^{*}) \leq \log(\frac{H}{n}).
\end{equation}
Finally note that again using the definition of $w_j^*$,
\begin{equation}
    \label{supAchievement}
    \begin{split}
    \sum_{j=1}^{n} w_j^{*} f_j - \mathrm{KL}(\sum_{j=1}^{n} \delta_{Y_j} w_j^{*}, \sum_{i=1}^{n} \delta_{Y_j} \frac{1}{n}) = \sum_{j=1}^{n} w_j^* f_j - \sum_{j=1}^{n} w_j^{*} \log(n w_j^{*}) = & \\
    \sum_{j=1}^{n} w_j^{*} f_j - \sum_{j=1}^{n} w_j^{*} \left( f_j - \log(\frac{H}{n}) \right) = & \\
    \log(\frac{H}{n}).
    \end{split}
\end{equation}
Equations \ref{upperBoundonSup} and \ref{supAchievement} complete the proof.
\end{proof}

\minMaxFlip*
\begin{proof}
Using the dual expression for semi-discrete Optimal Transport (see for example \cite{peyre2019computational} Chapter 5), we have that
\begin{equation}
    \begin{split}
        \ell_{\lambda}(P_n,P) = \inf_{\pmb{w} \in \Sigma_n} \sup_{\pmb{g} \in \mathbb{R}^{n}}  \mathbb{E}_{X \sim P} \min_{j} \| X - Y_j \|_2^2 - g_j) + \sum_{j=1}^{n} w_j g_j + \frac{1}{\lambda} \mathrm{KL}(\sum_{j=1}^{n} \delta_{Y_j} w_j, \sum_{j=1}^{n} \delta_{Y_j} \frac{1}{n})
    \end{split}
\end{equation}
The expression is concave in $\pmb{g}$ and convex in $\pmb{w}$. Applying Sion's min-max theorem \citep{komiya1988elementary}, we can flip the min with max. This yields
\begin{equation}
    \label{dualConnectionViaSion}
    \begin{split}
    \ell_{\lambda}(P_n,P) = & \\
    \sup_{\pmb{g} \in \mathbb{R}^n}  \mathbb{E}_{X \sim P} (\min_{j} \| X - Y_j \|_2^2 - g_j) + \inf_{\pmb{w} \in \Sigma_n} \sum_{j=1}^{n} w_j  g_j + \frac{1}{\lambda} \mathrm{KL}(\sum_{j=1}^{n} \delta_{Y_j} w_j, \sum_{j=1}^{n} \delta_{Y_j} \frac{1}{n}) = & \\
    \sup_{\pmb{g} \in \mathbb{R}^n} \mathbb{E}_{X \sim P} (\min_{j} \| X - Y_j \|_2^2 - g_j) - \frac{1}{\lambda} \left(\sup_{\pmb{w} \in \Sigma_n} \sum_{j=1}^{n} w_j (- \lambda g_j) -  \mathrm{KL}(\sum_{j=1}^{n} \delta_{Y_j} w_j, \sum_{j=1}^{n} \delta_{Y_j} \frac{1}{n}) \right) = & \\
    \sup_{\pmb{g} \in \mathbb{R}^{n}}  \mathbb{E}_{X \sim P}  \left(  \min_{j} (\| X - Y_j \|_2^2 - g_j) - \frac{1}{\lambda} \log(\frac{1}{n} \sum_{t=1}^{n} \exp(-\lambda g_t) ) \right) = & \\
    \sup_{\pmb{g} \in \mathbb{R}^{n}} \mathbb{E}_{X \sim P} h_1(X,g).
    \end{split}
\end{equation}
where in the second last equality we used Lemma \ref{variationalIdentity}. This proves Equation \ref{dualExpectation}. 
By Equation \ref{dualConnectionViaSion} and Lemma \ref{lem:infoAboutOurObjective} part (1) and (3), we have that there exist $g^{*} \in \mathbb{R}^{n}$ achieving the suprema in Equation \ref{dualExpectation} and any such $g^{*}$ must satisfies the gradient zero condition, which is the set of $n$ constraints
\begin{equation}
    \label{fixedPointEquation}
    P(L_j(g^{*})) = \frac{\exp(-\lambda g_j^{*})}{\sum_{t=1}^{n} \exp(-\lambda g^{*}_t)} \text{ for } j \in [n].
\end{equation}
What is left to show is Equation \ref{optWeightVectorExtraction}. Defining for $j \in [n], w^{*}_j = \frac{\exp(-\lambda g_j^{*})}{\sum_{t=1}^{n} \exp(-\lambda g_t^{*})}$ and $\hat{Q}_{\lambda} = \sum_{t=1}^{n} w_t^{*} \delta_{Y_t}$, Lemma \ref{kantorovichsFunctional} (1-3) and Equation \ref{fixedPointEquation} implies that
\begin{equation}
    \label{defOfW2}
    W_2^2(P,\hat{Q}_{\lambda}) = \mathbb{E}_{X \sim P} \left(\min_{j \in [n]} \| X - Y_j \|_2^2 - g_j^{*}\right) + \sum_{t=1}^{n} g_j^{*} w_j^{*}.
\end{equation}
In particular,
\begin{equation}
    \begin{split}
     \frac{1}{\lambda}\mathrm{KL}(\hat{Q}_{\lambda},P_n) +  W_2^2(P,\hat{Q}_{\lambda}) = & \\
     \mathbb{E}_{X \sim P} \left(\min_{j \in [n]} \| X - Y_j \|_2^2 - g_j^{*}\right) - \sum_{t=1}^{n} w_j^{*} (-  g_j^{*}) +  \frac{1}{\lambda} \mathrm{KL}(\hat{Q}_{\lambda},P_n)  =& \\
     \mathbb{E}_{X \sim P} \left(\min_{j \in [n]} \| X - Y_j \|_2^2 - g_j^{*}\right) -  \frac{1}{\lambda}\left(  \sum_{t=1}^{n} w_j^{*} (-\lambda g_j^{*}) - \mathrm{KL}(\sum_{t=1}^{n} w_t^{*} \delta_{Y_t},\sum_{t=1}^{n} \frac{1}{n} \delta_{Y_t}) \right)  \overset{\text{ Lemma \ref{variationalIdentity}} }{=} & \\
     \mathbb{E}_{X \sim P} \left(\min_{j \in [n]} \| X - Y_j \|_2^2 - g_j^{*}\right) - \frac{1}{\lambda} \log\left( \sum_{t=1}^{n} \frac{1}{n} \exp(-\lambda g_t^{*}) \right) \overset{\text{Def of }h_1}{=} & \\
    \mathbb{E}_{X \sim P} h_1(X,g^{*}) = & \\
    \ell_{\lambda}(P_n,P).
    \end{split}
\end{equation}
where the last equality is by Equation \ref{dualExpectation} and the definition of $g^{*}$.

\end{proof}

\sgaConvergence*
\begin{proof}

For an arbitrary $B >0$ to be specified later, we let $\gamma_{t} = B \sqrt{\frac{n}{t}}$ for $t \geq 1$. Additionally, let $g_{t-1}$ be the value of dual potential at the end of iteration $t$. By Theorem \ref{thm:minMaxFlip}, $\ell_{\lambda}(P_n,P) = \max_{g \in \mathbb{R}^{n}}  \mathbb{E}_{X \sim P} h_1(X,g)$. Applying Lemma \ref{lem:infoAboutOurObjective} part (4), there exists a $g^{*} \in \mathbb{R}^{n}$ achieving the maximum such that $\| g^{*} \|_{\infty} \leq D^2$. Now let $X_{t} \sim P$ be the sample from $P$ at the $t^{th}$ iteration of the algorithm. By Lemma \ref{helper:subDifferential}, $h_1$ is concave in its second argument. In particular, with $P$ probability $1$, for any subgradient in the second argument of $h_1$, denoted $\nabla_{g} h_1(X_{t},\cdot)$, we have
\begin{equation}
    \label{concavityUsage}
    h_1(X_{t},g_{t-1})+\nabla_{g} h_1(X_t,g_{t-1})^T(g^{*} - g_{t-1}) \geq h_1(X_t,g^{*}).
\end{equation}
In particular,
\begin{equation}
    \label{concavityUsage2}
    h_1(X_{t},g_{t-1})-h_1(X_{t},g^{*}) \geq \nabla_{g} h_1(X_{t},g_{t-1})^T(g_{t-1}-g^{*}).
\end{equation}
Taking the expectation on both sides and using the definition of $g^{*}$ and dividing by $\lambda$, we have
\begin{equation}
    \label{concavityUsage3}
    0 \geq \mathbb{E}_{X_t \sim P} h_1(X_{t},g_{t-1}) -  \ell_{\lambda}(P_n,P) \geq  \mathbb{E}_{X_t \sim P} \nabla_{g} h_1(X_t,g_{t-1})^{T}(g_{t-1}-g^{*}).
\end{equation}
By Lemma \ref{helper:subDifferential} and since $P$ is absolutely continuous with respect to Lebesgue measure, with $P$ probability $1$,
\begin{equation}
    g_{t} = g_{t-1}+\gamma_{t}  \nabla_{g} h_1(X_t,g_{t-1})
\end{equation}
for some subgradient $\nabla_{g} h_1$ (which is explicitly given in the algorithm), where $\gamma_{t}$ is the learning rate at iteration $t$. Expanding the square, we have that for $t \geq 1$,
\begin{equation}
    \label{gradientRecursion}
     \mathbb{E}_{P} \| g_t - g^{*} \|_2^2 = \mathbb{E}_{P} \| g_{t-1} - g^{*} \|_2^2 + 2 \gamma_t  \mathbb{E}_{X_{t} \sim P} \nabla_g h_1(X_t,g_{t-1})^T(g_{t-1}-g^{*})+ \mathbb{E}_{X_t \sim P} \gamma_t^2  \| \nabla_g h_1(X_t,g_{t-1}) \|_2^2.
 \end{equation}
 With $P$ probability $1$, the subgradient used in defining $g_{t}$ is the difference of two probability vectors. Hence,
 \[
 \mathbb{E}_{X_t \sim P} \gamma_{t}^{2}  \| \nabla_{g} h_1(X_t,g_{t-1}) \|_2^2 \leq 4 \gamma_t^2.
 \]
 Using this, and Equations \ref{concavityUsage3} and \ref{gradientRecursion}, we have that
 \begin{equation}
    \mathbb{E}_{P} \| g_t - g^{*} \|_2^2 - \mathbb{E}_{P} \| g_{t-1} - g^{*} \|_2^2 \leq 2 \gamma_{t} \left( \mathbb{E}_{X_t \sim P} h_1(X_t,g_{t-1}) -  \ell_{\lambda}(P_n,P) \right) + 4 \gamma_t^2.
 \end{equation}
 Rearranging terms yields, for $t \geq 1$,
 \begin{equation}
     0 \leq \gamma_{t} \left( \ell_{\lambda}(P_n,P) -  \mathbb{E}_{X_t \sim P} h_1(X_{t},g_{t-1})  \right) \leq \frac{1}{2} \left( \mathbb{E}_{P} \| g_{t-1} - g^{*} \|_2^2 - \mathbb{E}_{P} \| g_{t} - g^{*} \|_2^2 \right) + 2 \gamma_{t}^{2}. 
 \end{equation}

 Summing the above inequality over $t \in \{1,2,\dots,s\}$ for $s \geq 1$ and then dividing by $\sum_{t=1}^{s} \gamma_t$ yields
 \begin{equation}
     \label{penultimateEqn}
     0 \leq \mathbb{E}_{P} \left(  \ell_{\lambda}(P_n,P) - \frac{1}{\sum_{t=1}^{s} \gamma_t} \sum_{t=1}^{s} \gamma_t  h_1(X_t,g_{t-1}) \right) \leq \frac{\| g_0 - g^{*} \| _2^2 }{2  \sum_{t=1}^{s} \gamma_t} + 2 \frac{\sum_{t=1}^{s} \gamma_t^2}{\sum_{t=1}^{s} \gamma_s}.
 \end{equation}
 Now note that
 \begin{equation}
    \label{learningRateBoundA}
 \sum_{t=1}^{s} \gamma_t \geq s \gamma_s =  B \sqrt{n s}
 \end{equation}
 and
 \begin{equation}
    \label{learningRateBoundB}
 \sum_{t=1}^{s} \gamma_t^2 = nB^2 \sum_{t=1}^{s} \frac{1}{t} \leq n B^2(1+\sum_{t=2}^{s} \frac{1}{t}) \leq n B^2 (1+\int_{1}^{s} t^{-1} dt) = n B^2 (1+\log(s)).
 \end{equation}
 Using these and equation \ref{penultimateEqn}, we have that
 \begin{equation}
     0 \leq \mathbb{E}_{P} \left(  \ell_{\lambda}(P_n,P) - \frac{1}{\sum_{s=1}^{t} \gamma_t} \sum_{t=1}^{s} \gamma_t  h_1(X_t,g_{t-1}) \right) \leq (ns)^{-1/2} \left( \frac{1}{2B} \| g_0 - g^{*} \|_2^2 +2nB(1+\log(s)) \right) 
 \end{equation}
 Finally recall that $\| g^{*} \|_{\infty} \leq D^2$ and $g_0 = \mathbf{0}_{n}$. Hence we conclude, 
 \begin{equation}
    \label{biasBound}
 0 \leq \mathbb{E}_{P} \left(  \ell_{\lambda}(P_n,P) - \frac{1}{\sum_{t=1}^{s} \gamma_t} \sum_{t=1}^{s} \gamma_t  h_1(X_t,g_{t-1}) \right) \leq \sqrt{\frac{n}{s}} \left( \frac{D^{4}}{2B} +2B(1+\log(s)) \right) 
 \end{equation}

Now recalling the definition of $O$ (Equation \ref{defOfO}), we have that $\mathbb{E}_{P} h_1(X_{t},g_{t-1}) = \mathbb{E}_{P} \mathbb{E}_{P} \left(h_1(X_{t},g_{t-1}) | X_1,\dots,X_{t-1}\right) = \mathbb{E}_{P} O(g_{t-1})$. Using this, and that by Theorem \ref{thm:minMaxFlip}, $\ell_{\lambda}(P_n,P) \geq O(g)$ for $g \in \mathbb{R}^{n}$, we have
 \begin{equation}
     \begin{split}
         \mathbb{E}_{P}  \left| \ell_{\lambda}(P_n,P) - \frac{1}{\sum_{t=1}^{s} \gamma_t} \sum_{t=1}^{s} \gamma_t O(g_{t-1}) \right| = & \\
         \mathbb{E}_{P}  \left( \ell_{\lambda}(P_n,P) - \frac{1}{\sum_{t=1}^{s} \gamma_t} \sum_{t=1}^{s} \gamma_t O(g_{t-1}) \right) = & \\
         \ell_{\lambda}(P_n,P) - \frac{1}{\sum_{t=1}^{s} \gamma_t} \sum_{t=1}^{s} \gamma_t \mathbb{E}_{P} O(g_{t-1}) = & \\
         \ell_{\lambda}(P_n,P) - \frac{1}{\sum_{t=1}^{s} \gamma_t} \sum_{t=1}^{s} \gamma_t  \mathbb{E}_{P} h_1(X_t,g_{t-1}) = & \\
         \mathbb{E}_{P} \left(\ell_{\lambda}(P_n,P) - \frac{1}{\sum_{t=1}^{s} \gamma_t} \sum_{t=1}^{s} \gamma_t h_1(X_t,g_{t-1}) \right).
     \end{split}
 \end{equation}
Using this and Equation \ref{biasBound}, we have that
\begin{equation}
    \label{biasBound2}
    \mathbb{E}_{P}  \left| \ell_{\lambda}(P_n,P) - \frac{1}{\sum_{t=1}^{s} \gamma_t} \sum_{t=1}^{s} \gamma_t O(g_{t-1}) \right| \leq \sqrt{\frac{n}{s}} \left( \frac{D^{4}}{2B} +2B(1+\log(s)) \right).
\end{equation}

The next step is to produce an upper bound on $Var_{P} \frac{1}{\sum_{t=1}^{s} \gamma_t} \sum_{t=1}^{s} F_t$ with $F_t := \gamma_t (O(g_{t-1}) - h_1(X_t,g_{t-1}))$ for $t \geq 1$. Immediately we have that $\mathbb{E}_{P} \left(F_t | X_1,\dots,X_{t-1}\right) = 0$ (by definition of $O$ (equation \ref{defOfO})). Moreover, 
\begin{equation}
    \label{FtVariancePart1}
    Var_{P} F_t = \mathbb{E}_{P} Var_{P}(F_t | X_1,\dots,X_{t-1}) + Var_{P} E_{P} (F_t | X_1,\dots,X_{t-1}) = \mathbb{E}_{P} \gamma_t^2 Var_{P}( h_1(X_{t},g_{t-1}) | X_1,\dots,X_{t-1})
\end{equation}
 For every $g \in \mathbb{R}^{n}$ and $t \geq 1$,
 \begin{equation}
    \label{simpleVarianceUpperBound}
     Var_{P} h_1(X_t,g)  \leq \mathbb{E}_{P} \left( \left( h_1(X_t,g) - h_1(\mathbb{E}_{X \sim P} X,g) \right)^2 \right)
 \end{equation}
 where the above equation uses that the function $f(c) = \mathbb{E}_{P} (h_1(X_t,g) - c)^2$ is minimized with the choice $c = \mathbb{E}_{P} h_1(X_t,g)$. Now applying Lemma \ref{lipchitzH}, which is the $3D$-Lipchitz continuity of $h_1$ in its first argument on $A$, for any second argument input, we have that for $t \geq 1$, $g \in \mathbb{R}^{n}$
 \begin{equation}
     Var_{P} \left( h_1(X_t,g) \right) \leq 9 D^2 \mathbb{E}_{P} \| X_t - \mathbb{E}_{X \sim P} X \|_2^2 \leq 9  D^{4}.
 \end{equation}
 Using this and Equation \ref{FtVariancePart1}, we have that
 \begin{equation}
     \label{FtVariancePart2}
     Var_{P} F_t \leq \gamma_t^2 9 D^4.
 \end{equation}
 Using this and that we previously established $\mathbb{E}_{P} F_t | X_1,\dots,X_{t-1} = 0$ for $t \geq 1$ (which implies the summands are uncorrelated), we have that
 \begin{equation}
    \begin{split}
    \label{varianceBound}
    \mathbb{E}_{P} \left( \frac{1}{\sum_{t=1}^{s} \gamma_t} \sum_{t=1}^{s} F_t \right)^2 = & \\
     Var_{P} \frac{1}{\sum_{t=1}^{s} \gamma_t} \sum_{t=1}^{s} F_t = & \\
     \left(\frac{1}{\sum_{t=1}^{s} \gamma_t} \right)^2 \sum_{t=1}^{s}  Var_{P} F_t \leq & \\
     \frac{\sum_{t=1}^{s} \gamma_t^2 }{(\sum_{t=1}^{s} \gamma_t)^2} 9D^4
     \leq & \\
     9D^4\frac{n B^2 (1+\log(s))}{B^2 ns}
     \end{split}
 \end{equation}
 where in the last inequality we used Equations \ref{learningRateBoundA} and \ref{learningRateBoundB}. Finally, using both the bias bound (Equation \ref{biasBound2}) and the variance bound (Equation \ref{varianceBound}) and Jensen's inequality and that $\frac{1}{\sum_{t=1}^{s} \gamma_t} \sum_{t=1}^{s} \gamma_t h_1(X_t,g_{t-1}) = \hat{x}_{s,g_0,B}$,
 \begin{equation}
    \begin{split}
     \mathbb{E}_{P} |  \ell_{\lambda}(P_n,P) - \hat{x}_{s,g_0,B}| \leq \mathbb{E}_{P}|  \ell_{\lambda}(P_n,P) - \frac{1}{\sum_{t=1}^{s} \gamma_t} \sum_{t=1}^{s} \gamma_t O(g_{t-1}) | +  \mathbb{E}_{P} | \frac{1}{\sum_{t=1}^{s} \gamma_t}  \sum_{t=1}^{s} \gamma_t O(g_{t-1}) - \hat{x}_{s,g_0,B} | \leq & \\
     |  \ell_{\lambda}(P_n,P) - \frac{1}{\sum_{t=1}^{s} \gamma_t} \sum_{t=1}^{s} \gamma_t O(g_{t-1}) | + \sqrt{ \mathbb{E}_{P} | \frac{1}{\sum_{t=1}^{s} \gamma_t} \sum_{t=1}^{s} \gamma_t \left( O(g_{t-1}) - h_1(X_t,g_{t-1}) \right) |^{2} } \leq & \\
     \sqrt{\frac{n}{s}} \left( \frac{D^4}{2B} + 2B(1+\log(s))\right) + \sqrt{\mathbb{E}_{P} \left( \left( \frac{1}{\sum_{t=1}^{s} \gamma_t} \sum_{t=1}^{s} F_t \right)^2 \right) }\leq & \\
     \sqrt{\frac{n}{s}} \left( \frac{D^4}{2B} + 2B(1+\log(s))\right) + 3D^2 \sqrt{\frac{1+\log(s)}{s}}
     \end{split}
 \end{equation}
 Setting $B=1$ and $B = D^2$ yields the final result in Equation \ref{sgaPrecision}. The algorithm runtime is immediate as sub-gradient computation and $h_1$ evaluation takes $O(n)$ computation time, and there are $s$ iterations.
 
\end{proof}

\section{Full experimental settings and additional experimental results}
\label{sec:experimentalDetails}

The code for all experiments in this paper, in addition to detailed documentation for the code, can be found in the Supplemental zip file.

Table \ref{tab:param_settings} gives the parameter settings that we consistently use across all experiments in this paper. The only parameters that are experiment specific are: $\theta^{0}$ (the initial location for the parameter), and the upper and lower bounds for each parameter passed to CMA-ES (which vary across statistical model because the parameter spaces are different across statistical model). Also note that when estimating the function value using SGA Algorithm \ref{alg:SGA}, in practice it helps to wait a certain percentage of iterations before starting to accumulate the running average that is ultimately returned. For all experiments, we do this, only accumulating the running average for the last $40\%$ of iterations.
\begin{table}[h]
\caption{Consistent settings used across all experiments}
\centering
\begin{tabular}{lll}
\hline
\textbf{Variable} & \textbf{Description} & \textbf{Setting} \\
\hline
$M'$ & Number of bootstrap samples per $\lambda$ during $\lambda$ selection & 15 \\
$M$ & Number of bootstrap samples for primary inference (after $\lambda$ selected) & 100 \\
$s$ & Number of iterations per SGA & 20000 \\
$B$ & Learning Rate Scaling in SGA & 1 \\
$g_0$ & initial start vector in SGA & $\mathbf{0}_n$ \\
$R$ & Number of rounds per CMA-ES Optimization & 50 \\
$K$ & Population size per round of CMA-ES & 16 \\
$\sigma_0$ & initial step size of CMA-ES & 1 \\
\hline
\end{tabular}
\label{tab:param_settings}
\end{table}

\subsection{Univariate Normal experiments (including figure \ref{fig:lamSelect})}
\label{subsec:univariateNormalExperiments}

We run three experiments with the univariate normal statistical model, which is parameterized by the mean and the variance. The sample size is $n = 1000$ in all experiments. 

For comparison against minimum Wasserstein-2, we implement the strategy suggested by \cite{bernton2019approximate}. The quantity $W_2(P_n,P_{\theta})$ is approximated by drawing $x_1$ independent samples of size $x_2$ from $P_{\theta}$. This produces $x_1$ empirical measures of $P_{\theta}$ and the estimate of $W_2(P_n,P_{\theta})$ is the average over the $x_1$ values $W_2(P_n,\hat{P}_{\theta,j})$ where $\hat{P}_{\theta,j}$ is the $j^{th}$ empirical measure of $P_{\theta}$ (from $x_2$ samples of $P_{\theta}$). As in \cite{bernton2019approximate}, Nelder-Mead \citep{singer2009nelder} is used for optimization in the parameter space. We set $x_1 = 20$ and $x_2 = 20000$ as is done in \cite{bernton2019approximate}.

In each of the examples included in this section we do the following three things.(1): Run $\lambda$ selection to identify the elbow (and hence select a $\lambda$). Call it $\lambda^{*}$. (2): Given the choice $\lambda^{*}$, collect $M = 100$ bootstrap samples of the MRSW estimator using algorithm \ref{alg:singleBootstrapSample}; in addition we also collect 100 bootstrap samples from the Minimum Wasserstein 2 Estimator. (3): Using $\lambda^{*}$, run MRSW on the empirical measure of the data itself (which corresponds to running Algorithm \ref{alg:singleBootstrapSample} but with the empirical measure rather than constructing a bootstrap sample of the empirical measure), from which we extract $\hat{Q}_{\lambda}(\hat{\theta}_{\lambda})$ which is the reweighting of the empirical measure. 

In each of the experiments in this section, the figure corresponding to the experiment shows on the LEFT: the result of applying the $\lambda$ selection procedure explained in Section \ref{sec:lamSelect}. MIDDLE: the fitted densities corresponding to the bootstrap samples for B-MRSW vs the bootstrap samples from minimum Wasserstein-2. RIGHT: The reweighting of the raw sample achieved via $\hat{Q}_{\lambda}(\hat{\theta}_{\lambda})$.

For our method, we set the initial $\mu^{0} = -5$ and $\sigma^{0} = 0.15$. The parameter bounds for our method are $(-10,10)$ for $\mu$ and $(.10,20)$ for $\sigma$.

\subsubsection{Student T for geometric contamination. Dirac measure for Huber contamination}
\label{subsec:InPaperFullPipelineExample}
$P = (.95)t_{\nu = 22} + (.05) \delta_{10}$ where $t_{\nu}$ refers to the student T distribution with $22$ degrees of freedom. 

This is the experiment displayed as figure \ref{fig:lamSelect} in the main body of the text. It is redisplayed here for convenience as figure \ref{fig:tGeomContam}. The elbow is around $\lambda = 2.5$.

\begin{figure}
    \includegraphics[width=\textwidth]{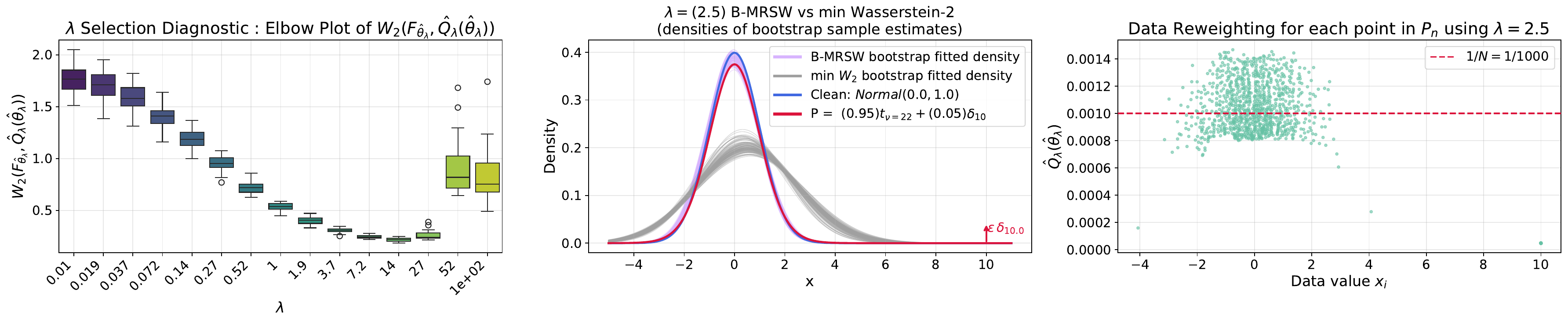}
    \caption{See Section \ref{subsec:InPaperFullPipelineExample} for full details}
    \label{fig:tGeomContam}
    
\end{figure}

\subsubsection{Discretization error for geometric contamination, biased Normal for Huber contamination}
\label{subsec:discretizationErrorNormal}

$P = (.95) T_{\rho} \# Normal(0,1) + (.05) Normal(8,1)$ where $T_{\rho}(x) = \lfloor \frac{x}{\rho} \rfloor \rho$ where $\#$ is the usual push-forward notation. Here we set $\rho = 0.05$. See Figure \ref{fig:discretizationErrorNormal}. Note that $\lambda$ selection results boxplots are similar to the previous experiment except for a general global decrease in the diagnostic value. The elbow is again around $\lambda = 2.5$.

\begin{figure}
    \includegraphics[width=\textwidth]{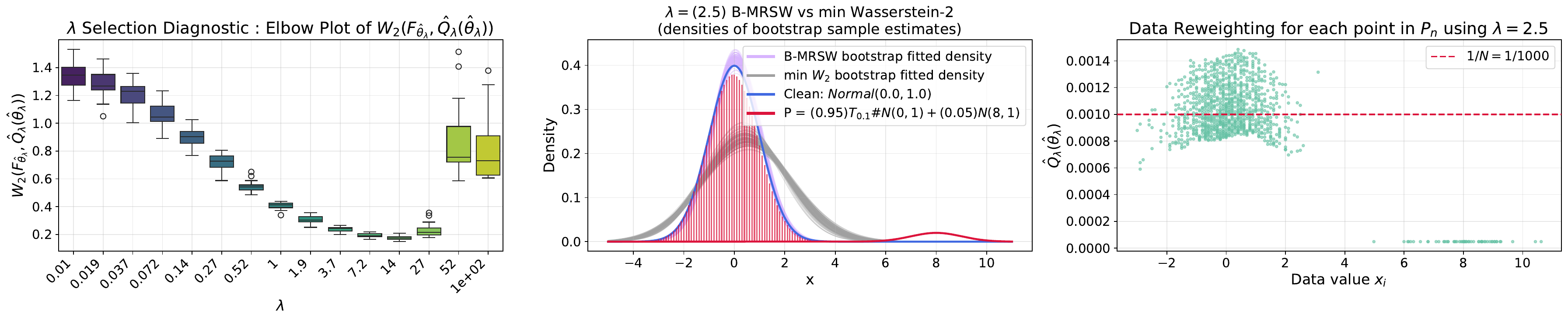}
    \caption{See Section \ref{subsec:discretizationErrorNormal} for details} \label{fig:discretizationErrorNormal}
\end{figure}

\subsubsection{Clean (no Huber or geometric contamination)}
\label{subsec:normalClean}
This example is included to demonstrate that the $\lambda$ selection exercise works properly even if there is no contaminant. See Figure \ref{fig:cleanNormal}. Here we see no elbow, indicating $\lambda \approx 0$ is reasonable to use. We choose $\lambda = .001$. Note that the gray and purple bands overlap each other, and the gray dominates the purple, which is why it appears there are only density estimates for minimum Wasserstein-2. But the two methods overlap completely and both do well, which makes sense given there is no contamination. 

\begin{figure}
    \includegraphics[width=\textwidth]{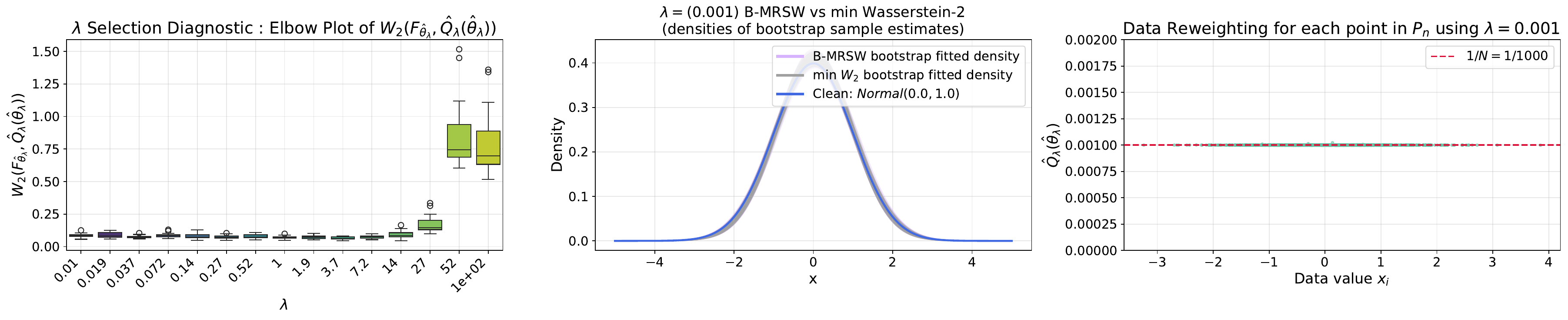}
    \caption{See Section \ref{subsec:normalClean} for details} \label{fig:cleanNormal}
\end{figure}



\subsection{g-and-k experiment}
\label{subsec:gandkExperimentDetails}

The contaminated data generating distribution for the experiment is $P = (1-\epsilon) Q + \epsilon \delta_{z}$ where $Y \sim Q$ if $X \sim GandK(a = 3, b = 1,g = 2, k =0.5)$ and $Y = \lfloor \frac{X}{\rho} \rfloor \rho$. We set $\rho = \epsilon =.05$ and $z = 50$.

We set the initial value to $(a^{0},b^{0},g^{0},k^{0}) = (5,.15,.05,.05)$ with bounds for $a$ of $(-10,10)$, for $b$ of $(.1,10)$, for $g$, $(.03,40)$ and $k$, $(.05,3.0)$.

For the NPL-MMD method, we use their code at \verb|https://github.com/haritadell/npl_mmd_project|. We use the default learning rate in their code base of $\gamma = 1$ and the default batch size in their code of $200$ -- (each iteration of their gradient descent procedure uses \verb|batch_size| samples from the empirical measure to represent the empirical measure). Since our total noise generation is $s (M'+M) = 115 * (20,000) = 2.3$ Million and the total noise generation cost in their method is $t \times  I \times M$ where $t$ is the number of noise vectors from $P_{\theta}$ per iteration of $SGD$ and $I$ is the number of iterations of their SGD, we set $t = 20$ and $I = 1000$. With $M = 100$ bootstrap samples, the total noise generation cost is $2$ million simulations which given the scale of the noise generation size is effectively equivalent to our noise generation budget. For more discussion regarding simulation efficiency, see subsection \ref{subsec:simEfficiencyComparison}.

The results of the $\lambda$ selection algorithm for our method are presented in Figure \ref{fig:gandkCleanvsContam} panel B. The elbow is in the approximate region $[0.5,3.5]$. To investigate sensitivity to the choice of $\lambda$ in the elbow, we show the results of inference using a range of $\lambda$ in the elbow region. Since \cite{dellaporta2022robust} arbitrarily uses bandwidth equal to $0.15$, we explore bandwidths near to this value, as well as the common Median Heuristic \citep{garreau2017large}. We compare both methods at sample sizes $n \in \{1000,5000\}$. For each sample size, we generate $20$ independent datasets. For each dataset, we compute the marginal bootstrap median estimate of each parameter using a given method, as well as compute a marginal confidence interval for each parameter. This produces $20$ marginal median estimates per parameter and $20$ confidence intervals per parameter. From the median estimates an average MSE to the ground truth parameter is computed. From the confidence intervals, coverage rates are computed. The complete results are presented in Table \ref{tab:gandk_comparison}. The marginal Median MSEs are further visualized in Figure \ref{fig:gandkCleanvsContam} panels C and D. Likewise, in Figure \ref{fig:gandkCleanvsContam} panels E and F the confidence interval coverages are visualized . Coverage rates for NPL-MMD are brittle for $g$ as a function of bandwidth. For B-MRSW, coverage rates are high across $\lambda$ except for kurtosis at $n = 5000$, for which our coverage greatly improves as $n$ increases. Coupled with the low average MSE for Kurtosis (across all $\lambda$) observed in Panel C, the under-coverage is a result of a finite sample bias correcting itself as $n$ increases. 


\begin{table}[h]
\caption{Comparison of G-and-K Inference Methods: Coverage, Median Interval Width (Wid), and MSE of Posterior Median per parameter}
\centering
\resizebox{\textwidth}{!}{
\begin{tabular}{llcc | cc | cc | cc | cc}
\hline
N & Method & $\lambda$ & Bandwidth & \multicolumn{2}{c|}{$a$} & \multicolumn{2}{c|}{$b$} & \multicolumn{2}{c|}{$g$} & \multicolumn{2}{c}{$k$} \\
 & & & & Cov/Wid & MSE & Cov/Wid & MSE & Cov/Wid & MSE & Cov/Wid & MSE \\
\hline
1000 & NPL-MMD & -- & Med. Heuristic & 0.70 (0.30) & 0.0136 & 1.00 (0.70) & 0.0240 & 0.00 (1.42) & 27.5786 & 0.95 (0.73) & 0.0248 \\
1000 & NPL-MMD & -- & 0.15 & 1.00 (0.48) & 0.0078 & 1.00 (1.07) & 0.0051 & 1.00 (6.40) & 0.3086 & 1.00 (3.69) & 0.0705 \\
1000 & NPL-MMD & -- & 0.3 & 1.00 (0.39) & 0.0075 & 1.00 (0.81) & 0.0150 & 0.90 (6.16) & 1.8782 & 1.00 (1.48) & 0.0326 \\
1000 & NPL-MMD & -- & 0.5 & 0.40 (0.28) & 0.0220 & 1.00 (0.78) & 0.0238 & 0.00 (2.56) & 30.2764 & 0.95 (1.13) & 0.0510 \\
1000 & B-MRSW & 0.5 & -- & 0.95 (0.20) & 0.0025 & 0.85 (0.39) & 0.0155 & 1.00 (0.89) & 0.0338 & 0.35 (0.24) & 0.0209 \\
1000 & B-MRSW & 1.5 & -- & 0.85 (0.17) & 0.0026 & 0.90 (0.37) & 0.0081 & 0.85 (0.73) & 0.0672 & 0.20 (0.23) & 0.0236 \\
1000 & B-MRSW & 2.5 & -- & 0.90 (0.17) & 0.0026 & 1.00 (0.36) & 0.0055 & 0.60 (0.70) & 0.1004 & 0.15 (0.24) & 0.0258 \\
1000 & B-MRSW & 3.5 & -- & 0.85 (0.17) & 0.0027 & 1.00 (0.35) & 0.0048 & 0.55 (0.63) & 0.1216 & 0.15 (0.22) & 0.0263 \\
\noalign{\hrule height 1.5pt}
5000 & NPL-MMD & -- & Med. Heuristic & 0.30 (0.26) & 0.0171 & 1.00 (0.56) & 0.0222 & 0.00 (1.07) & 28.7610 & 1.00 (0.60) & 0.0153 \\
5000 & NPL-MMD & -- & 0.15 & 1.00 (0.42) & 0.0041 & 1.00 (0.59) & 0.0044 & 0.80 (3.98) & 0.1407 & 0.90 (2.27) & 0.0528 \\
5000 & NPL-MMD & -- & 0.3 & 1.00 (0.33) & 0.0042 & 1.00 (0.54) & 0.0101 & 0.75 (5.69) & 0.2369 & 1.00 (0.79) & 0.0156 \\
5000 & NPL-MMD & -- & 0.5 & 0.20 (0.24) & 0.0232 & 1.00 (0.58) & 0.0162 & 0.00 (2.04) & 32.5825 & 1.00 (0.66) & 0.0337 \\
5000 & B-MRSW & 0.5 & -- & 1.00 (0.12) & 0.0005 & 0.75 (0.26) & 0.0063 & 1.00 (1.04) & 0.0048 & 0.75 (0.18) & 0.0043 \\
5000 & B-MRSW & 1.5 & -- & 0.90 (0.09) & 0.0007 & 1.00 (0.18) & 0.0010 & 1.00 (0.36) & 0.0100 & 0.70 (0.16) & 0.0038 \\
5000 & B-MRSW & 2.5 & -- & 0.90 (0.09) & 0.0009 & 1.00 (0.20) & 0.0010 & 1.00 (0.35) & 0.0139 & 0.60 (0.17) & 0.0047 \\
5000 & B-MRSW & 3.5 & -- & 0.90 (0.09) & 0.0009 & 1.00 (0.20) & 0.0010 & 0.90 (0.36) & 0.0202 & 0.45 (0.18) & 0.0056 \\
\hline
\end{tabular}
}
\label{tab:gandk_comparison}

\caption{Comparison of G-and-K Inference Methods: Coverage, Median Interval Width (MedW), and MSE of Posterior Median Estimator. $20$ unique datasets are generated according to the contaminated data generating distribution with $n \in \{1000,5000\}$. Each method uses $100$ bootstrap samples per dataset. The marginal coverage rates per parameter, and median interval width are reported. The final column considers point estimation, where the bootstrap samples are used to construct the vector of marginal medians. The Mean MSE of the vector of marginal medians to the ground-truth parameters of the clean distribution are reported.}
\label{tab:gandk_comparison}
\end{table}

\subsection{Regarding simulation efficiency}
\label{subsec:simEfficiencyComparison}

NPL-MMD is inherently not robust even to Huber Contamination when the bandwidth is selected poorly (see Section \ref{sec:MMDnotRobust}). However, given a number of noise generations $Z$ to solve an optimization in parameter space, NPL-MMD is more efficient in its total number of calls to the simulator relative to B-MRSW. In performing a single stochastic gradient descent in parameter space, NPL-MMD  only calls the simulator once per noise, whereas in CMA-ES for optimization in parameter space, B-MRSW calls the simulator (\verb|CMA-ES Rounds|)*(\verb|CMA-ES population size|) times per noise.

B-MRSW uses many more queries to the simulator because it is attempting to perform global optimization, rather than just finding a critical point. NPL-MMD on the other hand is only seeking out a critical point in its objective function landscape. Comparing the number of queries to the simulator between B-MRSW and NPL-MMD is thus not a fair comparison because NPL-MMD does not guarantee that with a sufficiently high number of queries to the simulator, the optimum in the loss landscape will actually be found.

But we again emphasize that even if NPL-MMD performed a global optimization, with a poorly chosen bandwidth inferences are not robust to Huber contamination (see Section \ref{sec:MMDnotRobust} for additional details), and no systematic procedure for bandwidth selection to achieve robust inferences is provided by \cite{dellaporta2022robust}.

\subsection{Regarding experiment runtime}
\label{subsec:wallClockTime}

On an 8 core, 2024 Macbook Air with M3 chip and 16GB of memory the runtime is 3 minutes and 32 seconds to perform one CMA-ES optimization using the parameter settings in Table \ref{tab:param_settings} when $n = 5000$ data points are from the contaminated G-and-K distribution of Section \ref{sec:numericalIllustrations} and the G-and-K simulator is used. Constructing a confidence interval based on $M = 100$ bootstrap samples thus takes approximately 350 minutes. However, one can leverage that the bootstrap samples can be collected in parallel. We used Google Cloud Services; using 100 N2 standard 4CPU machines in parallel, the runtime to process 100 bootstrap samples reduces to 9 - 10 minutes. 

A single optimization using NPL-MMD at the settings described in Section \ref{subsec:gandkExperimentDetails} takes only 50 seconds on the 8 core, 2024 Macbook Air with M3 chip and 16GB of memory (for the $N = 5000$ points from the contaminated G-and-K distribution when using the G-and-K simulator). However, again we note that their Stochastic Gradient Descent algorithm can only guarantee convergence to a critical point in the objective landscape, whereas CMA-ES targets a global optimizer in the objective function landscape.

\section{The limits of Gaussian kernel MMD for bandwidth selection}
\label{sec:MMDnotRobust}

This section establishes that Gaussian-kernel MMD is non-robust at the high bandwidth extreme, helping to explain the bandwidth sensitivity of NPL-MMD observed in Section \ref{sec:numericalIllustrations}.

The Gaussian MMD is defined for two distributions $P$ and $Q$ as
\begin{equation}
    \label{gaussianMMDdefinition}
    \begin{split}
    MMD^2(P,Q) = \mathbb{E}_{X \sim P, Y \sim P} \exp(-\frac{1}{2 \sigma_0^2} (X-Y)^2 ) + \mathbb{E}_{X \sim Q, Y \sim Q} \exp(-\frac{1}{2 \sigma_0^2} (X-Y)^2) - & \\
    2 \mathbb{E}_{X \sim P, Y \sim Q} \exp(-\frac{1}{2 \sigma_0^2} (X-Y)^2)
    \end{split}
\end{equation}


\subsection{Gaussian MMD behaves like a first moment distance for large bandwidth when dimension = 1}
\label{MMDnotRobust}

\begin{lemma}
Suppose $P$ and $Q$ have finite fourth moments ($\mathbb{E}_{X \sim P} X^4 < \infty$, $\mathbb{E}_{Y \sim Q} Y^4 < \infty$). Then
\[
\lim_{\sigma_0 \to \infty} \sigma_0^2 MMD^2(P,Q) = (\mathbb{E}_{X \sim P} X - \mathbb{E}_{X \sim Q} X)^2
\]
\end{lemma}
\begin{proof}
The two term taylor series expansion of the Kernel is
\begin{equation}
\label{taylorForMMD}
\exp(-\frac{1}{2 \sigma_0^2} (x-y)^2) - \exp(0) = -\frac{1}{2 \sigma_0^2} (x-y)^2 + \frac{\exp(\xi(x,y)) \left( \frac{1}{2 \sigma_0^2} (x-y)^2 \right)^2}{2}
\end{equation}
where $\xi(x,y) \in \left[  -\frac{1}{2\sigma_0^2} (x-y)^2, 0 \right]$. By the above, we have the pointwise limit, for every $x,y$, as $\sigma_0 \to \infty$
\[
\sigma_0^2 \left(\exp(-\frac{1}{2 \sigma_0^2} (x-y)^2 ) -1\right) \to -\frac{1}{2} (x-y)^2
\]
By Equation \ref{taylorForMMD}, and the finite fourth moment condition, we also have for any $\sigma_0 > 1$
\[
| \sigma_0^2 \left( \exp(-\frac{1}{2 \sigma_0^2} (x-y)^2) -1\right) | \leq \frac{(x-y)^2}{2} + \frac{(x - y)^{4}}{4}
\]
By the finite fourth moment condition, the right hand side function above is $L_1( P \times P), L_1(P \times Q), L_1(Q \times Q)$. By the definition of $MMD^2$ and DCT we thus have
\[
\sigma_0^2 MMD^2(P,Q) =
\]
\begin{equation}
    \begin{split}
     \sigma_0^2 \mathbb{E}_{X \sim P, Y \sim P} \left( \exp(-\frac{1}{2 \sigma_0^2} (X-Y)^2) -1\right)+ \sigma_0^2 \mathbb{E}_{X \sim Q, Y \sim Q} \left( \exp(-\frac{1}{2 \sigma_0^2} (X-Y)^2) -1\right) - & \\
    2 \sigma_0^2 \mathbb{E}_{X \sim P, Y \sim Q} \left( \exp(-\frac{1}{2 \sigma_0^2} (X-Y)^2)-1\right) \to & \\
    -\frac{1}{2} \mathbb{E}_{X \sim P, Y \sim P} (X-Y)^2 - \frac{1}{2} \mathbb{E}_{X \sim Q, Y \sim Q} (X-Y)^2 + \mathbb{E}_{X \sim P, Y \sim Q} (X-Y)^2 = & \\
    - Var(P) - Var(Q) + \mathbb{E}_{X \sim P, Y \sim Q} (X-Y)^2 = & \\
    -Var(P)-Var(Q) + Var(P)+Var(Q) + (\mathbb{E}_{X \sim P} X - \mathbb{E}_ {X \sim Q} X)^2 = & \\
    (\mathbb{E}_{X \sim P} X - \mathbb{E}_{X \sim Q} X)^2
    \end{split}
\end{equation}

\end{proof}

\newpage

\end{document}